\documentclass[journal]{IEEEtran}

\def\DIFS    {\mbox{\rm DIFS}}
\def\SIFS    {\mbox{\rm SIFS}}

\usepackage{wrapfig,enumerate}

  \usepackage[dvips]{graphicx}
  \graphicspath{{./Figures/}}
  \DeclareGraphicsExtensions{.eps}

\usepackage[cmex10]{amsmath}
\usepackage{amssymb,comment}

\usepackage[tight,footnotesize]{subfigure}

\newtheorem{theorem}{Theorem}

\begin{document}


\title{Modeling, Analysis and Impact of a Long Transitory Phase in Random Access Protocols}

\author{Cristina Cano and David Malone
\thanks{C. Cano and D. Malone are with the Hamilton Institute,
National University of Ireland, Maynooth,
Co. Kildare, Ireland.
E-mails: \{cristina.cano,david.malone\}@nuim.ie}
}

\maketitle

\begin{abstract}

In random access protocols, the service rate depends on the number of stations with a packet buffered for transmission. We demonstrate via numerical analysis that this state-dependent rate along with the consideration of Poisson traffic and infinite (or large enough to be considered infinite) buffer size may cause a high-throughput and extremely long (in the order of hours) transitory phase when traffic arrivals are right above the stability limit. We also perform an experimental evaluation to provide further insight into the characterisation of this transitory phase of the network by analysing statistical properties of its duration. The identification of the presence as well as the characterisation of this behaviour is crucial to avoid misprediction, which has a significant potential impact on network performance and optimisation. Furthermore, we discuss practical implications of this finding and propose a distributed and low-complexity mechanism to keep the network operating in the high-throughput phase. 

\end{abstract}

\begin{IEEEkeywords}
Stability, random access protocols, mean field analysis, decoupling approximation, DCF, Aloha, Homeplug.
\end{IEEEkeywords}

\IEEEpeerreviewmaketitle

\section{Introduction}\label{sec:intro}


A common characteristic of random access protocols is the provision of a service rate that is dependent on the actual number of stations with a packet pending for transmission (backlogged stations). We show in this work that this state-dependent service rate, in combination with exponentially distributed packet interarrivals, may cause an extremely long transitory period (of the order of magnitude of hours) under certain conditions. In particular, when the queue length of the stations is large enough to be considered infinite and we operate right after the stability limit of the network. Consider, for example, a network formed by $50$ nodes contending for the channel using Homeplug 1.0 Medium Access Control (MAC) \cite{cano2013PLCmodel}. Fig. \ref{fig:instantaneous_tp_homeplug} shows the instantaneous throughput (measured every second). Observe the long time during which the network remains in a high-throughput phase and how the behaviour of throughput suddenly changes.


To illustrate why this effect takes place, consider a set of nodes with no previous packets buffered for transmission generating packets at a rate slightly higher than the maximum rate the network could serve in saturation. In this situation, the probability that a large percentage of the nodes contend for the channel at the same time is small. Thus, compared to the case in which all the stations are backlogged, the time to transmit a packet is reduced as the conditional collision probability is smaller. Consequently, the probability that a large number of packets accumulate for transmission is also reduced and higher throughput than that achieved when all stations are saturated can be obtained. However, after the stability limit of the network and under infinite queue size, this situation cannot be maintained in the long run. Eventually, the number of backlogged stations increases and so does the time to transmit a packet, leading to a consecutive increase of the queued packets and moving the network to 
saturation (the stable operating point). We show in this work that the time it takes the network to reach the long-term stability can be extremely long, though this depends on, among other parameters, the value of the packet generation rate.


We build upon previous literature on mean field and queue stability analysis of random access protocols and provide a numerical evaluation to demonstrate the potential extremely long duration of this transitory phase as well as the conditions under which it may occur. Moreover, we provide more insight into the characterisation of its duration by means of experimental evaluations. As far as we know, this is the first work that demonstrates and characterises this transitory phase, crucial to avoid misprediction of performance results. The specific contributions are as follows:

\begin{figure}[tb] 
\centering
\includegraphics[width=2.345in]{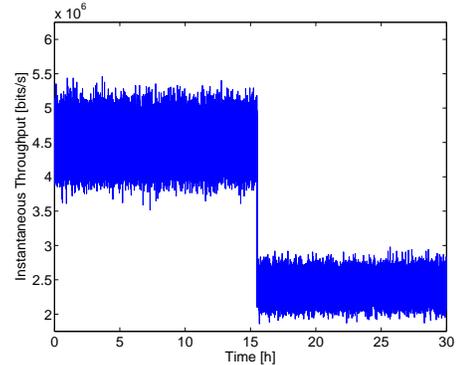} 
\caption{Evolution of instantaneous throughput in Homeplug MAC from \cite{cano2013PLCmodel}.}
\label{fig:instantaneous_tp_homeplug}
\end{figure}

\begin{enumerate}
 \item A demonstration via numerical analysis that in a coupled system of queues of length long enough to be considered infinite there may be a long transitory phase when the packet arrival rate is slightly higher than the maximum load that the system could serve in saturation. We also perform an extensive experimental evaluation to support this finding.
 \item Further insight into the statistical properties of the duration of this transitory phase.
 \item We contribute further understanding of random access protocols with the aim of avoiding misprediction of results that have a significant potential impact on performance evaluation, parametrisation and optimisation.
 \item We discuss practical implications of our findings and provide a simple and distributed mechanism to keep the network in the high-throughput phase, which provides substantial gains in network performance. 
\end{enumerate}

The remainder of this article is organised as follows. In Section \ref{sec:motivation}, we motivate this work by providing an overview of related work as well as previous results affected by misprediction. Then, in Section \ref{sec:identifying}, we build upon previous literature to demonstrate that the high-throughput phase is transitory and describe the challenges involved in its detection, both analytical and experimental. After that, we describe the methodology used in this work to both \emph{i)} demonstrate the transitory phase may have a potentially long duration and \emph{ii)} provide more insight into its statistical properties. Results are presented in Section \ref{sec:results}. Then, a discussion on practical implications, including a mechanism to keep the network operating in the transitory phase, is provided in Section \ref{sec:implications}. Finally, we conclude the paper with some final remarks.

\section{Motivation}\label{sec:motivation}


We have detected the long transitory period in a range of random access protocols, including Aloha and the Distributed Coordination Function (DCF) defined in IEEE 802.11 \cite{IEEE80211-IEEESTD1999} (as will be shown in Section \ref{sec:results}) as well as the Homeplug channel access procedure \cite{HomeplugStd} (see Fig. \ref{fig:instantaneous_tp_homeplug} and also \cite{cano2013PLCmodel}). A common feature of these protocols is that the time to transmit a packet depends on the actual number of contending stations (those stations that are associated to the network and that have a packet pending for transmission). When the traffic arrival rate is just above the level at which we expect the network to be stable, we find that this variability in service rate combined with the variability in traffic arrivals can lead to a high-throughput long-transitory phase before the queues saturate. However, under infinite, or large enough to be considered infinite, queue size, the queues will 
eventually 
become unstable. In these conditions, a larger number of stations contend for the channel reducing throughput to that found in saturation. The possibility of a long transitory phase for random access protocols in this regime has already been postulated \cite{suleiman2008impact}, however without experimental findings or proof. The high-throughput phase was also described in \cite{cao2010modeling}. The authors argued that the network is unstable in this regime and showed experimentally the change of behaviour. Here we go further and demonstrate that the phase that provides the saturation throughput is, in fact, the only stable operation of the network. Thus, the high-throughput phase corresponds to a transitory period. Moreover, we demonstrate and characterise its potentially long duration.


We first found the long transitory behaviour in Homeplug MAC \cite{cano2013PLCmodel}. In that work, we experimentally show the evolution of the instantaneous throughput and find that after a long period of time, the queues start to fill with packets and the throughput faces a sharp decrease to the saturation throughput. We also discuss that analytical models based on the common approximation of decoupling queue states and service times are not able to differentiate between the transitory and the stable solutions and that care should be taken when performing experimental evaluations (for completeness, we will briefly address these problems in Section \ref{sec:identifying}). One of the contributions in \cite{cano2013PLCmodel} is to show that previous results under these conditions \cite{chung2006performance} are incomplete, as their simulation and analytic model results correspond to the network performance of the transitory phase.

In this work we show that this effect is not Homeplug-specific but common to many random access protocols. We demonstrate that the high-throughput phase is transitory and that its duration may be extremely long (see Sections \ref{sec:identifying} and \ref{sec:results}, respectively). We also give more insight into the magnitude and distribution of the time it takes the network to reach the stable phase (see Section \ref{sec:results}). We believe these contributions are crucial to avoid misprediction of results, both from analysis and simulation, that have a potentially significant impact on performance evaluation and comparison, configuration of network channel access parameters and in optimisation analysis. Misprediction of results can lead to wrong conclusions with a clear impact on network performance, especially in this particular case, as the transitory is a higher-throughput phase compared to the stable solution. Moreover, the difference in throughput among the transitory and stable phases, depending 
on 
channel access parameters and 
conditions, can be extremely large (more than a $50\%$ 
reduction in some scenarios). These findings make the detection of such misprediction in performance evaluation even more relevant.

A clear example of such misprediction in IEEE 802.11 can be found in \cite{pitts2008analysing}, where an analytical model based on the decoupling approximation and infinite queue size is proposed. Higher throughput than the one in saturation is predicted when, for the conditions specified, $50$ and $100$ nodes are contending for the channel and the packet arrival rate is approaching saturation. However, we can observe disagreement in that region when comparing their analytical and experimental results (see Fig. 6 in \cite{pitts2008analysing}). In particular, the throughput found from simulations for the higher packet arrival rates in that region corresponds to neither the saturated throughput nor the highest throughput phase, an effect we believe is caused by averaging results from the transitory phase and the stable operation. 

\section{Transitory and Stable Operation of the Network}\label{sec:identifying}

In this section, we first build upon previous literature to show that the lowest-throughput phase corresponds to the stationary operation of the network. Then, we discuss the challenges involved, both analytical and experimental, in identifying the long-term stable solution of the network. We provide techniques to overcome these challenges in order to obtain a valid performance prediction, also from analytical and experimental point of views. 

\subsection{Stability Analysis} 


Building upon the results in \cite{borst2008stability}, where the queue stability of a system of parallel, coupled queues with infinite buffer size is studied, we demonstrate that the stable (stationary) operation of the network corresponds to saturation, as follows.

\begin{theorem}

Let $X = (X_1, ..., X_N)$ be the queue length process of a system of $N$ parallel and homogeneous queues with infinite buffer sizes, strictly positive and Poisson-distributed arrival rates ($\lambda_i$, with $1 \leq i \leq N$ denoting a given queue) and bounded state-dependent service rates $\mu_i(x)$. If $\exists i$ that satisfies:

\begin{equation}\label{eq:condition1_proof}
\lambda_i>\limsup_{x_i \rightarrow \infty} \sup_{x_j:j\neq i} \mu_i(x).
\end{equation}

Then, independently of the initial state: 

\begin{equation}\label{eq:condition2_proof}
\min_i(x_i(t))_{t \rightarrow \infty} > 0.
\end{equation}

\end{theorem}
 
\begin{proof}Such a process is shown to be transient in \cite{borst2008stability} if Eq. \ref{eq:condition1_proof} is satisfied. Thus, it follows that $x_i(t)$ tends to go to infinity. Therefore, with $t \rightarrow \infty$ the number of packets waiting for transmission at queue $i$ will be higher than zero. Moreover, assuming homogeneity, having an unstable queue in the system implies that all the other queues are also unstable \cite{luo1999stability}. Thus, the condition in Eq. \ref{eq:condition2_proof} holds by applying the same proof for all $N$ queues. In such conditions the operation of the network corresponds to saturation (all nodes have at least a packet pending for transmission).
\end{proof}

\subsection{Identifying the Stable Operation: Analytical Challenges}\label{sec:identifying_decoupling}


When analysing the performance of a network formed by a set of nodes using a random access protocol, the decoupling approximation is commonly used in order to make the analysis tractable. Under the decoupling approximation each queue is modelled as independent of the rest of queues in the network. However, this assumption, although practical under a range of circumstances, can provide two different solutions for throughput in certain regimes \cite{Duffy2010}. In particular, under the conditions we have detected the long transitory phase: \emph{i)} when the packet arrival rate is slightly higher than the maximum load that the system could serve in saturation and \emph{ii)} under infinite (or large enough to be considered infinite) buffer size. These models involve solving a fixed point equation, however, in contrast to saturated models, the solution may not be unique in models that consider unsaturated conditions. 

The fundamental limitation of these analytical models is that they do not consider the number of instantaneous contending stations, i.e., the number of queues that have at least a packet buffered at the same time. Neglecting this fact makes it impossible to capture the actual regime in which the queues operate. This is caused by the possibility of facing two extreme cases: the queues being mostly empty or saturated conditions. 
This effect is in fact related to the transitory phase, the system operates in a high-throughput phase, that corresponds to one of the solutions the analytical model converges to, and then, moves to the saturated network operation, which corresponds to the other solution. The relation between the packet arrivals rates for which we observe two different solutions and so the potential for a long transitory 
phase is studied in this section.

For the two best-known random access protocols, Aloha and DCF, we illustrate here how analytical models based on the decoupling approximation may exhibit two solutions by varying the initial conditions of our iterative solver. Furthermore, we use these analytical models to obtain the service rate in saturation (the stability limit of the network), a metric that allows us to identify when there is a potential for the long transitory phase to take place. Using \emph{Theorem 1}, we define the stability condition as $\lambda < \mu_{\rm sat}$, where $\mu_{\rm sat}$ denotes the service rate when $N$ stations are contending for the channel. When the stability condition of the network is not satisfied, there is potential to both \emph{i)} obtain different solutions from the analysis and \emph{ii)} find a long transitory phase from one solution to the other. However, note that when $\lambda >> \mu_{\rm sat}$, we expect the probability of finding a transitory period to decrease as the probability of having a large 
number of the nodes simultaneously contending for the channel increases.

\subsubsection{Aloha}

We take a standard renewal reward approach \cite{kumar05,bianchi05} to model the network performance of an Aloha network. The analytical model used is described in Appendix \ref{appendix:aloha}. To make comparison with DCF easier, parameters shown in Table \ref{tbl:parameters} are used. The contention window ($W$) is set to $32$ and the data payload ($L$) is set to $1500$ bytes. Results for the aggregate throughput for different number of nodes ($N$) and packet arrival rates ($\lambda$), as well as the service rate in saturated conditions for different values of $N$, are depicted in Fig. \ref{fig:renewal_reward_aloha}. 

Observe that, when $\lambda$ is slightly higher than the maximum service rate in saturation for a given $N$ (shown in Fig. \ref{fig:renewal_reward_aloha_mu}), the analysis converges to two different solutions for throughput (Fig. \ref{fig:renewal_reward_aloha_tp}). The solution labelled as \emph{Analysis 1} is obtained by using initial conditions representing saturation for the numerical method, i.e.,: \emph{i)} number of idle slots between renewal events equal to $0$, \emph{ii)} queue occupancy equal to $1$, \emph{iii)} probability of finding the queue non empty after a transmission equal to $1$ and \emph{iv)} transmission attempt rate set to a high value ($0.5$). On the contrary, the solution labelled as \emph{Analysis 2} considers the opposite case (lightly-loaded initial conditions for the numerical method): \emph{i)} number of idle slots between renewal events equal to a large value ($1000$), \emph{ii)} queue occupancy equal to $0$, \emph{iii)} probability of finding the queue non empty after a transmission equal to $0$ 
and \emph{iv)} 
transmission 
attempt rate set to a small value ($1\cdot10^{-5}$). Note that both solutions represent a fixed point of the model.

The range of packet arrival rates for which we obtain the two solutions increases with the number of nodes in the network. The reasons behind this result are twofold. First, with a small number of nodes, the probability (independently of the service rate) that all of them have a packet for transmission at the same time is higher. Second, with higher $N$, there is a higher range of possible operating points before reaching the saturated service rate. Observe that, under these specific conditions, when $N=10$ there is no discrepancy among the two different solutions and thus, we do not expect the long transitory phase to occur. 

\begin{table}[tb!]
\centering
\caption{System parameters of the IEEE 802.11b specification \cite{IEEE80211-IEEESTD1999}.}\label{tbl:parameters}
\begin{tabular}{|c||c|} \hline
Parameter & Value in IEEE 802.11b\\ \hline
$R_{\rm data}$ & $11$~Mbps \\ \hline
$R_{\rm basic}$/$R_{\rm PHY}$ & $1$~Mbps \\ \hline
$L_{\rm MACH}$ & $272$~bits \\ \hline
$L_{\rm PLCPPre}$ & $144$~bits \\ \hline
$L_{\rm PLCPH}$ & $48$~bits \\\hline
$L_{\rm ack}$ & $112$~bits \\ \hline
$\sigma$ & $20$~$\mu$s \\ \hline
$\rm{DIFS}$ & $50$~$\mu$s \\ \hline
$\rm{SIFS}$ & $10$~$\mu$s \\ \hline
\end{tabular}
\end{table}

\begin{figure*}[!tb] 
\centering
\subfigure[Throughput]{\includegraphics[width=2.3in]{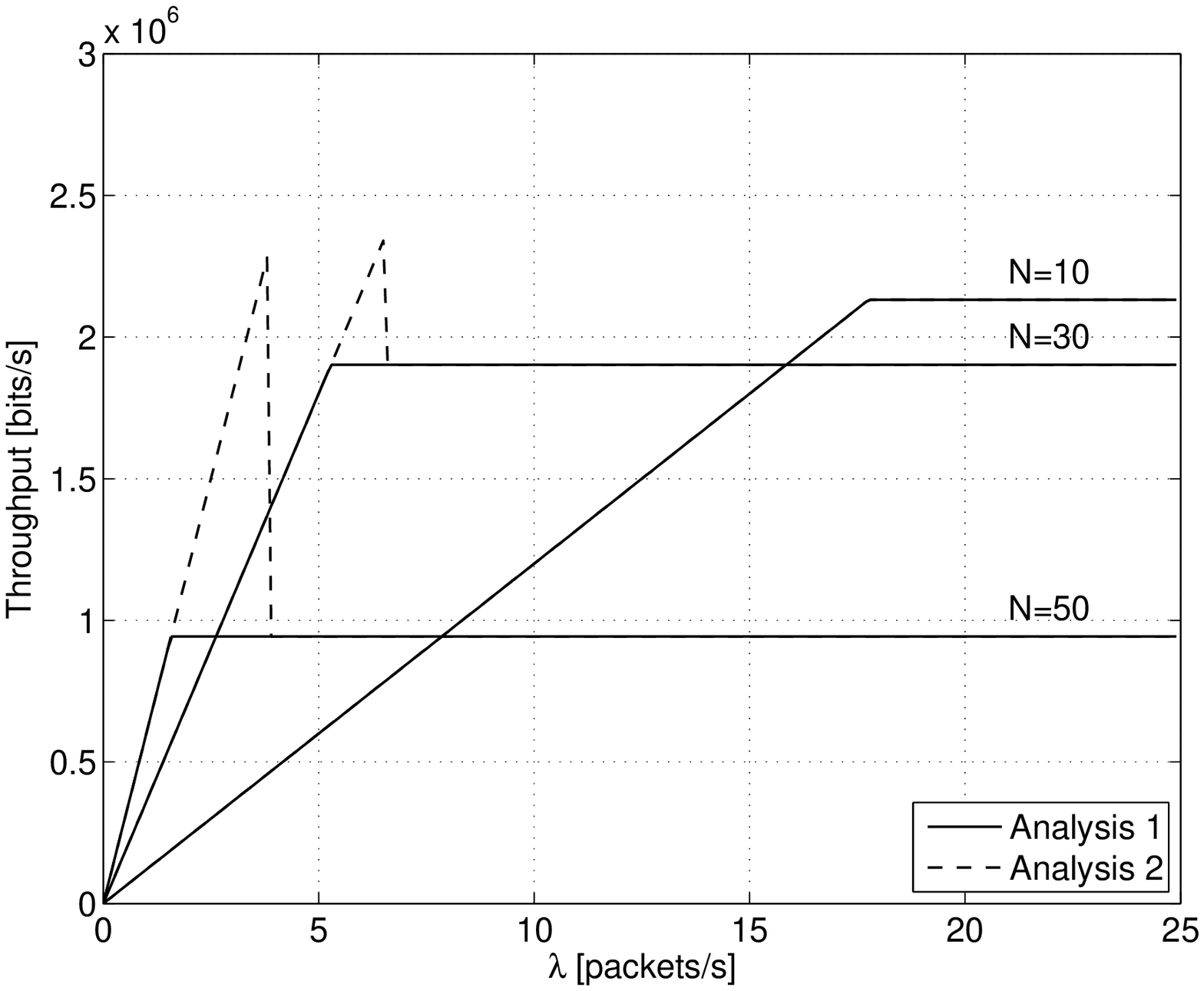}\label{fig:renewal_reward_aloha_tp}}
\subfigure[Saturated Service Rate]{\includegraphics[width=2.3in]{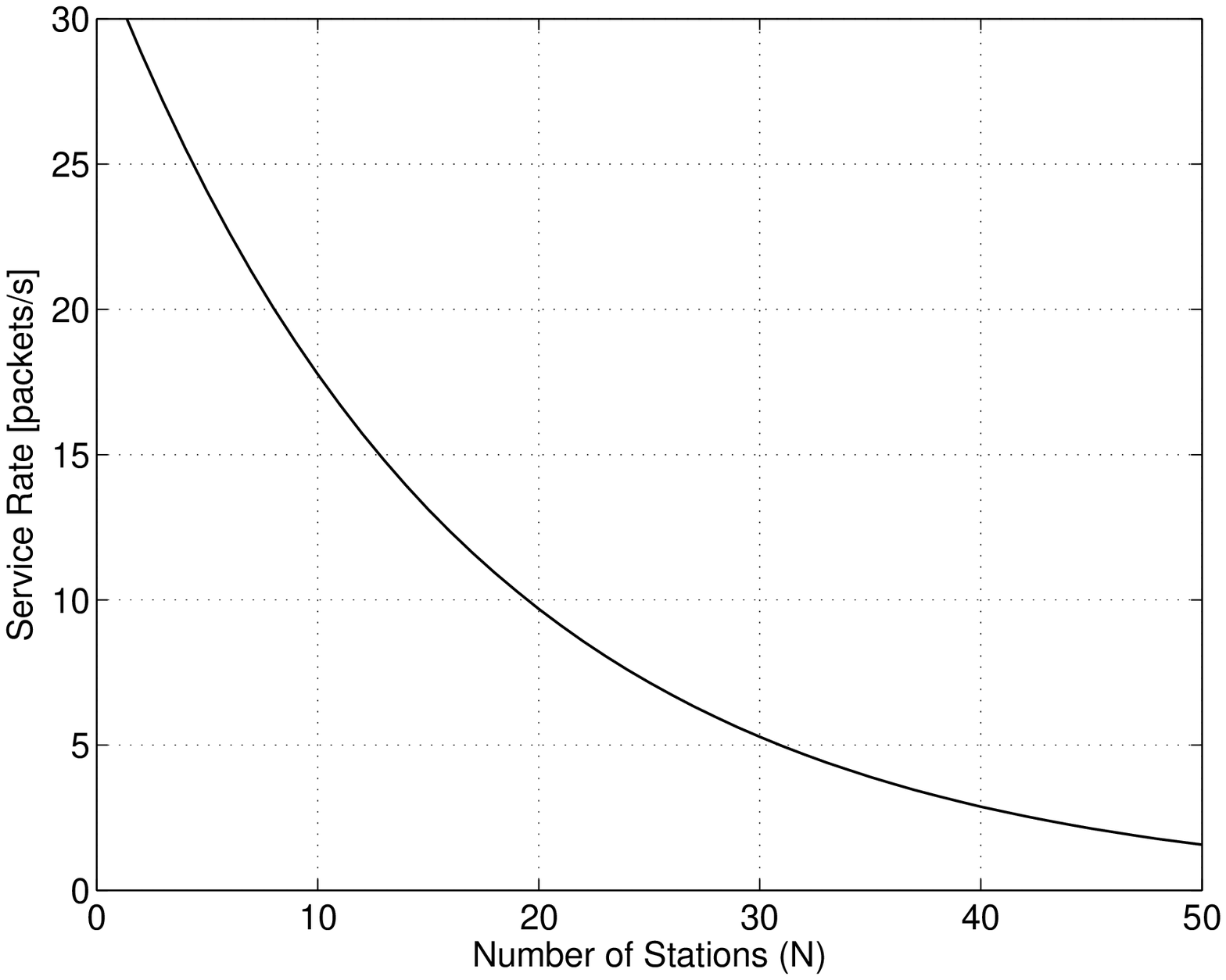}\label{fig:renewal_reward_aloha_mu}}
\caption{The two solutions obtained from the renewal reward analysis for the aggregated throughput and service rate in saturated conditions in Aloha.}
\label{fig:renewal_reward_aloha}
\end{figure*}

\begin{figure*}[!tb] 
\centering
\subfigure[Throughput ($W=8$, $m=3$)]{\includegraphics[width=2.3in]{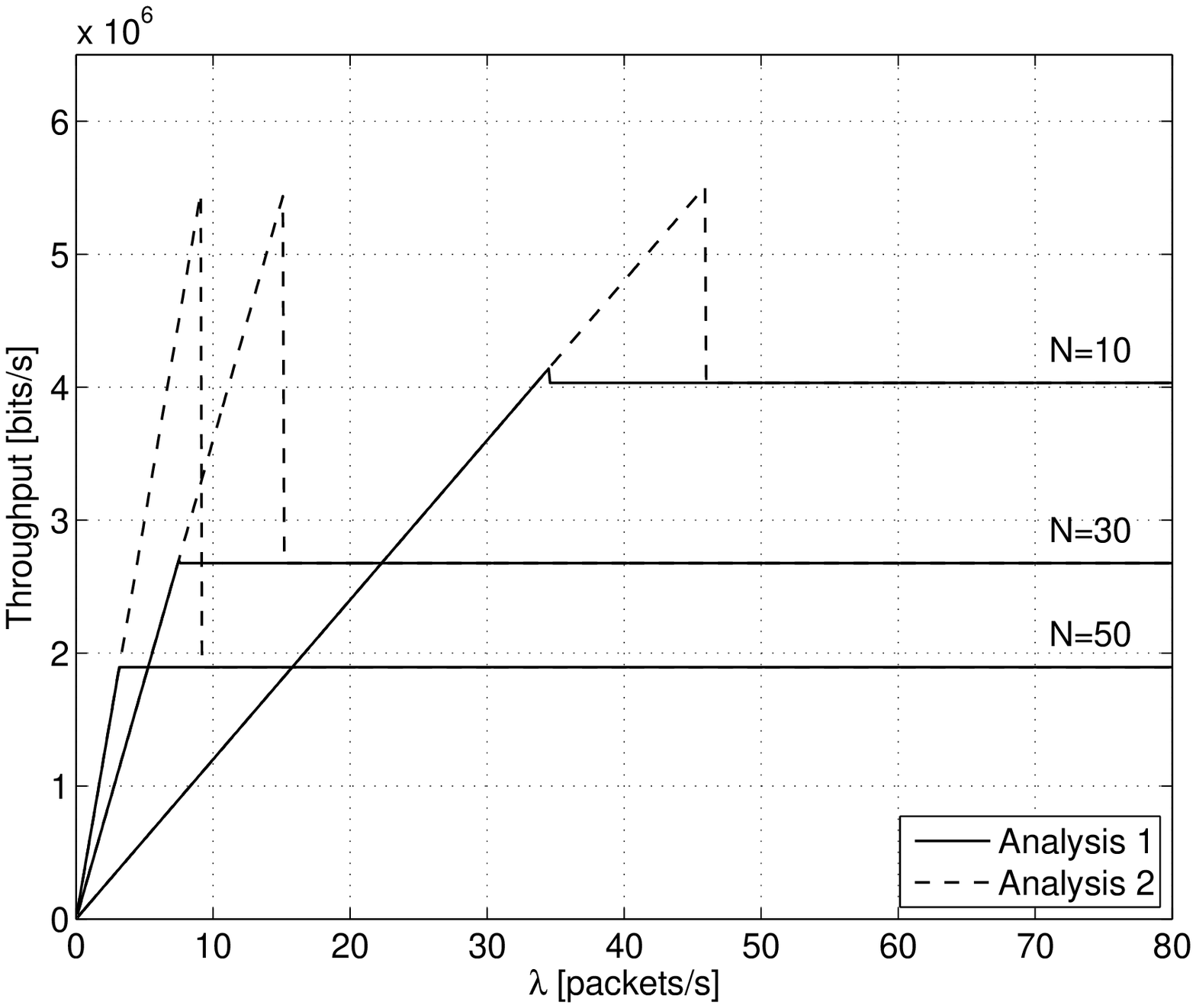}\label{fig:renewal_reward_dcf_tp_cw8_m3}}
\subfigure[Throughput ($W=32$, $m=5$)]{\includegraphics[width=2.3in]{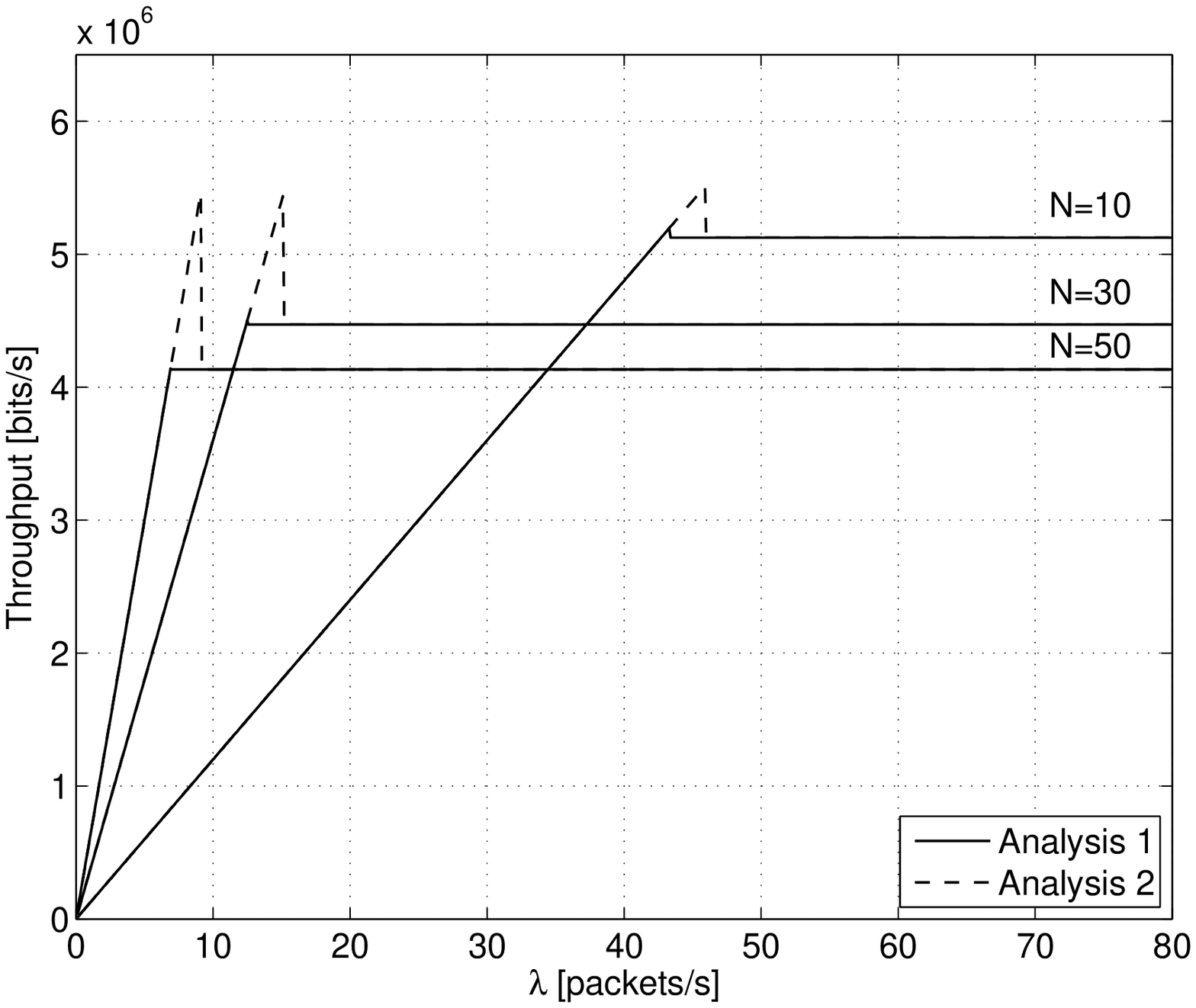}\label{fig:renewal_reward_dcf_tp_cw32_m5}}
\subfigure[Saturated Service Rate]{\includegraphics[width=2.4in]{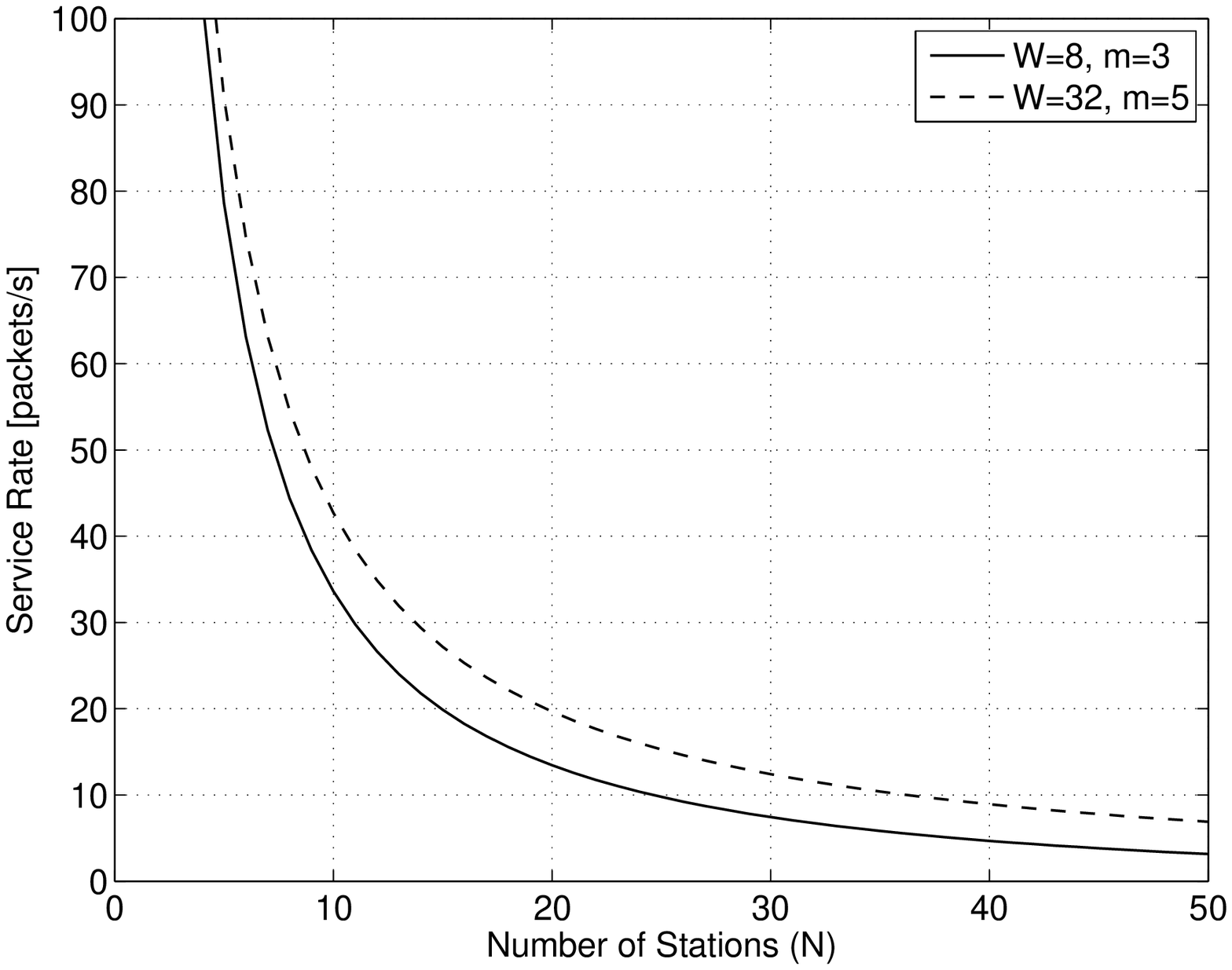}\label{fig:renewal_reward_dcf_mu}}
\caption{The two solutions obtained from the renewal reward analysis for the aggregated throughput and service rate in saturated conditions in DCF.}
\label{fig:renewal_reward_dcf}
\end{figure*}

\subsubsection{DCF}

We have also studied the effect of obtaining two different solutions from the analysis for a network with nodes using the DCF random access protocol. The renewal reward approach presented in Appendix \ref{appendix:dcf} has been used for this purpose. Fig. \ref{fig:renewal_reward_dcf} shows the aggregated throughput as well as the service rate in saturated conditions for two different combinations of the contention parameters. We have analysed the cases in which the minimum contention window ($W$) is set to $8$ and $32$ and the number of backoff stages ($m$) is equal to $3$ and $5$. Thus, with maximum contention windows being $64$ and $1024$, respectively (for full details of IEEE 802.11 DCF see \cite{IEEE80211-IEEESTD1999}). The different solutions labelled as \emph{Analysis 1} and \emph{Analysis 2} have been obtained as in the last subsection. 

Also, as found for Aloha, when the packet arrival rate is immediately higher than the service rate in saturation (depicted in Fig. \ref{fig:renewal_reward_dcf_mu}), there is a potential for obtaining two solutions from the analysis. We observe in this case that the range of packet arrival rates for which two solutions appear is higher with $W=8$ and $m=3$ than for $W=32$ and $m=5$ (Fig. \ref{fig:renewal_reward_dcf_tp_cw8_m3} and \ref{fig:renewal_reward_dcf_tp_cw32_m5}). This is caused by the larger difference between the maximum throughput that can be achieved with a reduced number of backlogged stations and the saturation throughput. Note how the saturation throughput with the smaller contention window and number of backoff stages is much lower than when the contention window is larger.

Observe also that, despite our efforts to choose the initial values for the solver to produce two solutions, only one solution is obtained when the two solutions are very close. This can be observed in Fig. \ref{fig:renewal_reward_dcf_tp_cw8_m3} and \ref{fig:renewal_reward_dcf_tp_cw32_m5} for $N=10$ when the packet arrival rate is close to the maximum service rate the system could serve in saturation. However, 
the existence of two different solutions for throughput in those conditions has been confirmed with the methodology used in \cite{Duffy2010}. 

In conclusion, it is crucial to know the two-solution existence when interpreting results from the analysis in order to avoid misprediction. In fact, this is the reason for the incorrect predicted long-term performance of IEEE 802.11 presented in \cite{pitts2008analysing} and the Homeplug MAC in \cite{chung2006performance}. The decoupling approximation along with a model exhibiting two solutions resulted in incomplete 
results from the 
analysis. Furthermore, in \cite{chung2006performance}, agreement was found comparing the performance obtained from the analysis with simulations, leading the authors to believe the analytical model was validated. Given the long transitory phase, care must be also taken when running an experimental evaluation as we will discuss later in this section. Thus, verifying whether the stability condition is satisfied comparing the maximum service rate (Figs. \ref{fig:renewal_reward_aloha_mu} and \ref{fig:renewal_reward_dcf_mu}) and the actual packet arrival rate ($\lambda$) is crucial to guarantee that outcomes correspond to long-term network operation.

\subsubsection{Previous Efforts on Modelling Coupled Dynamics}

Despite the vast majority of analytical models of random access protocols being based on the decoupling approximation, some authors have already pointed out the inaccuracy of this assumption and have made efforts on trying to model the coupled dynamics of the system of queues (e.g. \cite{garetto2005performance,panda2009state}). The main problem in considering the number of packets at each buffer and modelling the coupled behaviour remains in the large resulting state space. Keeping track of the number of packets buffered for transmission at each queue translates into a state space of $Q^N$, where $Q$ denotes the maximum length of the queues and $N$ the number of stations. Consequently, even for small queue lengths the complexity of the system is considerably higher than modelling the network assuming that the decoupling approximation holds. Furthermore, as $Q$ tends to infinity, the analysis becomes intractable. We will describe this issue in more detail in Section \ref{sec:methods}. 

A state-dependent model that considers the actual number of contending stations is proposed in \cite{garetto2005performance}. The analysis tracks the queue evolution of a tagged station and approximates the probability that any of the other non-tagged stations find themselves with an empty buffer after a given transmission. The approximation taken in that work considers that this probability depends only on the number of competing stations and is derived from the buffer occupancy probability of the tagged station (assuming the backoff stage of the non-tagged station differs at most by one). Based also on a state-dependent service rate approach, an analytical model with reduced complexity is proposed in \cite{panda2009state}. However, as opposed to \cite{garetto2005performance}, the stationary distribution is used to derive the probability that a departure from a non-tagged queue leaves the queue empty. 

In this work, we will also apply the state-dependent service rate approach to model the duration of the transitory period in Section \ref{sec:methods}. However, we will define a simpler approximation to compute the probability of stations having an empty queue after a packet transmission. As will be shown in Section \ref{sec:methods}, we consider the queue occupancy probability after a packet transmission to be only dependent on the current state (number of backlogged stations). This assumption will prove adequate to demonstrate the transitory phase being of long duration.

\subsection{Identifying the Stable Operation: Experimental Challenges}


When $\lambda >> \mu_{\rm sat}$, the system rapidly moves to the stable solution. However, performing an experimental evaluation right after the stability limit of the network can provide wrong results as the length of the transitory phase can be extremely long (see Fig. \ref{fig:instantaneous_tp_homeplug} in the introductory section). If the experiments are started with the queues empty, it can take a long time to reach the stable solution since the system has to reach a point at which a large number of nodes are simultaneously contending for the channel and start to have their queues filled with an increasing number of packets.  


One way to obtain the stable results is to start with the queues empty, run the experiments for a long time until the system changes to the stable solution and then start taking the statistics of the performance metrics of interest. Note that, due to the extremely long duration of the transient, the statistics from the transitory phase must be discarded in order to obtain the valid performance results even when setting a long time horizon. However, we suggest a more practical way to force the system to enter into the stable operation: to start the experiments with a number of packets preloaded in the queues. If the queues are unstable, we have started the experiments closer to the stationary regime, and will see the long-term throughput more quickly. This technique is based on the recommendation to set the initial conditions to those in steady-state proposed in \cite{wilson1978evaluation}. Method shown to be more effective in estimating the steady-state mean if compared to truncation (discarding the 
transients caused by initial conditions). The drawback of this technique is that, if the queues are stable, there is a transitory phase during which those extra packets are released. However, this technique proved useful in predicting the performance even in unsaturated conditions in \cite{cano2013PLCmodel}. Thus, demonstrating that, although care has to be taken when selecting the amount of packets to preload queues and the time horizon, nodes are able to release these packets in a reasonable amount of time.

\section{Methods and Metrics to Assess the Length of the Transitory Period}\label{sec:methods}

In this section, we describe the methods and metrics selected to get more insight into the duration and distribution of the transitory phase. For each method we will define which metrics of interest can be obtained in order to estimate the value of interest, i.e., the instant at which the transitory phase ends. Results will be presented in Section \ref{sec:results}.

\subsection{Evaluation Methods}

The evaluation methods considered in this work are: \emph{i)} a queueing system with coupled service rates, \emph{ii)} an analytic model of the number of backlogged stations and \emph{iii)} a network simulation framework. We next describe them in detail and discuss their complexity and accuracy.

\subsubsection{Method 1 - Queueing system with coupled service rates}

We model the system of $N$ parallel queues as a discrete Markov process in $\mathbb{Z}_{+}^{N}$ in which each state represents the number of packets waiting for transmission at each queue: $X = (X_1, ..., X_N)$. We assume the system of parallel queues to be homogeneous, i.e., same maximum queue length, packet arrival and average service rate at all queues. The number of backlogged stations in a given state $x$ is denoted by $n_x$ and represents the number of queues with at least one packet pending for transmission. We take into account exponentially distributed packet arrivals at rate $\lambda$ packets/s. The service rate ($\mu(n_x)$) is also considered to be exponentially distributed as well as state-dependent. Such dependence on the actual state models the impact of the number of backlogged stations ($n_x$) on performance metrics such as the conditional collision probability and/or the average backoff duration (depending on the access protocol in use). Note that the state-dependent service rate can 
be chosen to match the access protocol of interest (Aloha, DCF, Homeplug MAC, etc.). With all these considerations, the transition 
probabilities among the different states of this process can be defined as:

\begin{align}\label{eq:transitions}
 P(x \mapsto x+e_i \leq q)= & \frac{\lambda}{N\lambda + n_x\mu(n_x)},\nonumber\\
 P(x \mapsto x-e_i \geq 0)= & \frac{\mu(n_x)}{N\lambda + n_x\mu(n_x)},
\end{align}

with the relational operators being element-wise, $q$ and $e_i$ being the all $Q$ and $i$-th unit vectors in $\mathbb{Z}_{+}^{N}$, respectively. The number of backlogged stations in state $x$ is computed as:

\begin{equation}\label{eq:nx}
 n_x = \sum_{i=1}^{N} I(x_i),
\end{equation}

where $I(x_i)$ is the indication function of having at least one packet pending for transmission in queue $i$:

\begin{equation}\label{eq:indication}
 I(x_i) = \left\{ \begin{array}{rll}
1 &   \mbox{if} & x_i > 0, \\
0 &    \mbox{otherwise.} &
 \end{array}\right.
\end{equation}

Although this system provides a close description of the behaviour of the network, the complexity of solving it explicitly is prohibitive. Note that its state space is of the order of $Q^N$ and we are interested in the case in which $Q \to \infty$. Therefore, the system is computationally intractable even considering small $N$ and taking into account that the transition matrix is sparse. However, we can perform Monte Carlo simulations of this process in order to experimentally characterise the duration and distribution of the transitory phase. 

\subsubsection{Method 2 - Modelling the number of backlogged stations}\label{sec:methods_2}

With the goal of simplifying the system of coupled queues described in \emph{Method 1}, we consider now a discrete Markov process in $\mathbb{Z}_{+}$ in which every state represents the number of backlogged stations ($X=n_x$). We consider the queue occupancy probability after a packet transmission independent of the previous states. Thus, the probability that a station that has transmitted a packet remains still backlogged is modelled as the standard queue busy probability: $\rho_x = \lambda/\mu(x)$, where $x$ denotes the current state, i.e., the number of backlogged stations. Using this approximation, we are effectively ignoring the queue length and only considering whether the queue has a packet pending for transmission. Note that $\rho_x$ is, in fact, not memory-less. However, this assumption allows us to reduce the computational complexity of the previous analysis while still preserving information about the queue occupancies, crucial to obtain insight into the transitory phase.

We set the number of states of the Markov Chain to $1+N'$ (from having no backlogged station up to the case in which $N'$ stations are backlogged), where $N'$ is the smallest value of $N$ for which the condition $\lambda < \mu(N')$ is no longer satisfied. Observe that $\rho_x > 1$ when $N'$ stations have a packet pending for transmission.  Thus, we consider the last state to be absorbing, i.e., $P(N' \mapsto N')= 1$. The transition probabilities for $x<N'$ are shown in Eq. \ref{eq:transitions_approx}. 

\begin{align}\label{eq:transitions_approx}
 P(x \mapsto x+1 \leq N')= & \frac{(N-x)\lambda}{(N-x)\lambda + x\mu(x)},\nonumber\\
 P(x \mapsto x-1 \geq 0)= & \frac{x\mu(x)(1-\rho_x)}{(N-x)\lambda + x\mu(x)},\nonumber\\
 P(x \mapsto x)= & \frac{x\mu(x)\rho_x}{(N-x)\lambda + x\mu(x)}.
\end{align}

The state space and complexity of this system is significantly reduced compared to \emph{Method 1} as the need to keep track of the queue occupancy at every queue is removed, reducing the state space to $N'+1$. However, it only provides an approximation of the actual behaviour of the network.

\subsubsection{Method 3 - Network Simulations}

We have also used a network simulator based on the SENSE framework \cite{chen2004sense}. Packet interarrival times are modelled as exponentially distributed while the service rate strictly follows the DCF random access procedure. The results obtained from simulations are the closest to the real behaviour of the network as the assumptions considered are minimal. However, network simulations are time consuming and thus, they are impractical to derive conclusions for a large range of conditions. Network simulations are used in this work to evaluate the accuracy of the outcomes obtained from the previously described methods as well as to obtain results that closely match the actual behaviour of the network.

\subsubsection{Embedded Time vs. Real Time}

Note that from \emph{Method 1} and \emph{Method 2} we obtain number of events (packet arrivals and departures) while in \emph{Method 3} we measure time in seconds. However, we will use Gillespie's stochastic simulation algorithm (\emph{direct method}) \cite{gillespie1977exact} in \emph{Method 1} to compute at each packet arrival and departure the time interval to the next event. 

\subsection{Metrics}

The metrics we have defined aim to provide a close estimation of the time at which the system moves from the transitory phase to the stable behaviour. However, depending on the method, some limitations apply. We describe the different metrics and how they relate to the previously defined evaluation methods next.

\subsubsection{Metric 1 - Hitting Time of a Limiting State}\label{sec:methods_metric1}

Starting with the queues empty, we first consider the instant, on average, at which a limiting state (a state for which the stability condition is not satisfied) is first hit. This metric allows us to track the time elapsed since the network start-up until $N'$ stations are backlogged.



By means of Monte Carlo simulations of the system described in \emph{Method 1} we can easily compute the number of events (packet arrival/departures), on average, to hit one of the limiting states. On the contrary, using \emph{Method 2}, we can obtain the average number of events to hit state $x=N'$ starting from $x=0$ (denoted as $h(0)$) by solving the system of linear equations formed by Eq. \ref{eq:hitting_approx} along with  $h(N') = 0$.

\begin{align}\label{eq:hitting_approx}
\begin{split}
 h(x < N') = 1 & + \frac{(N-x)\lambda}{(N-x)\lambda + x\mu(x)} h(x+1) \\ & + \frac{x\mu(x)(1-\rho_x)}{(N-x)\lambda + x\mu(x)} h(x-1>0) \\ & + \frac{x\mu(x)\rho_x}{(N-x)\lambda + x\mu(x)} h(x). \nonumber\\ 
\end{split} \\ 
\end{align}

This analysis is extremely computationally efficient and thus, can be used to perform an extensive numerical evaluation. However, results are affected by the approximation taken removing the need to track the queue occupancies in the system as described in \emph{Method 2}. Moreover, observe that this metric allows us to get some insight into the duration of the transitory phase but also that having $N'$ backlogged stations does not guarantee that the system moves to the stable phase. If the number of packets in the queues is reduced, there is still some probability that the stations are able to transmit those packets without facing an increase in the number of packets waiting for transmission. Thus, keeping the system in the transitory phase. Therefore, the average number of events to hit state $N'$ is a lower bound of the events necessary to escape from the transitory period. 

To illustrate this fact, we track the queue evolution at every time instant (packet arrival/departure) using Monte Carlo simulations of the system described in \emph{Method 1}. The minimum, average and maximum queue length of $N=50$ nodes using the DCF protocol with $\lambda=7.5$ packets/s for two different simulation runs are depicted in Fig. \ref{fig:q_evolution}. Observe that the queue occupancies remain low during a long time interval until they start to increase to the maximum queue length. We plot the instant at which the limiting state (first $x$ such that $\lambda>\mu(n_x)$) is first reached (vertical line). Note that, in Fig. \ref{fig:q_evolution_1}, this instant coincides with the moment at which the queues start to be filled with packets. However, in Fig. \ref{fig:q_evolution_2}, the system is able to recover from this situation 
and remain 
in the transitory phase for a longer time interval. 

\begin{figure}[!tb] 
\centering
\subfigure[Simulation Run 1]{\includegraphics[width=2.5in]{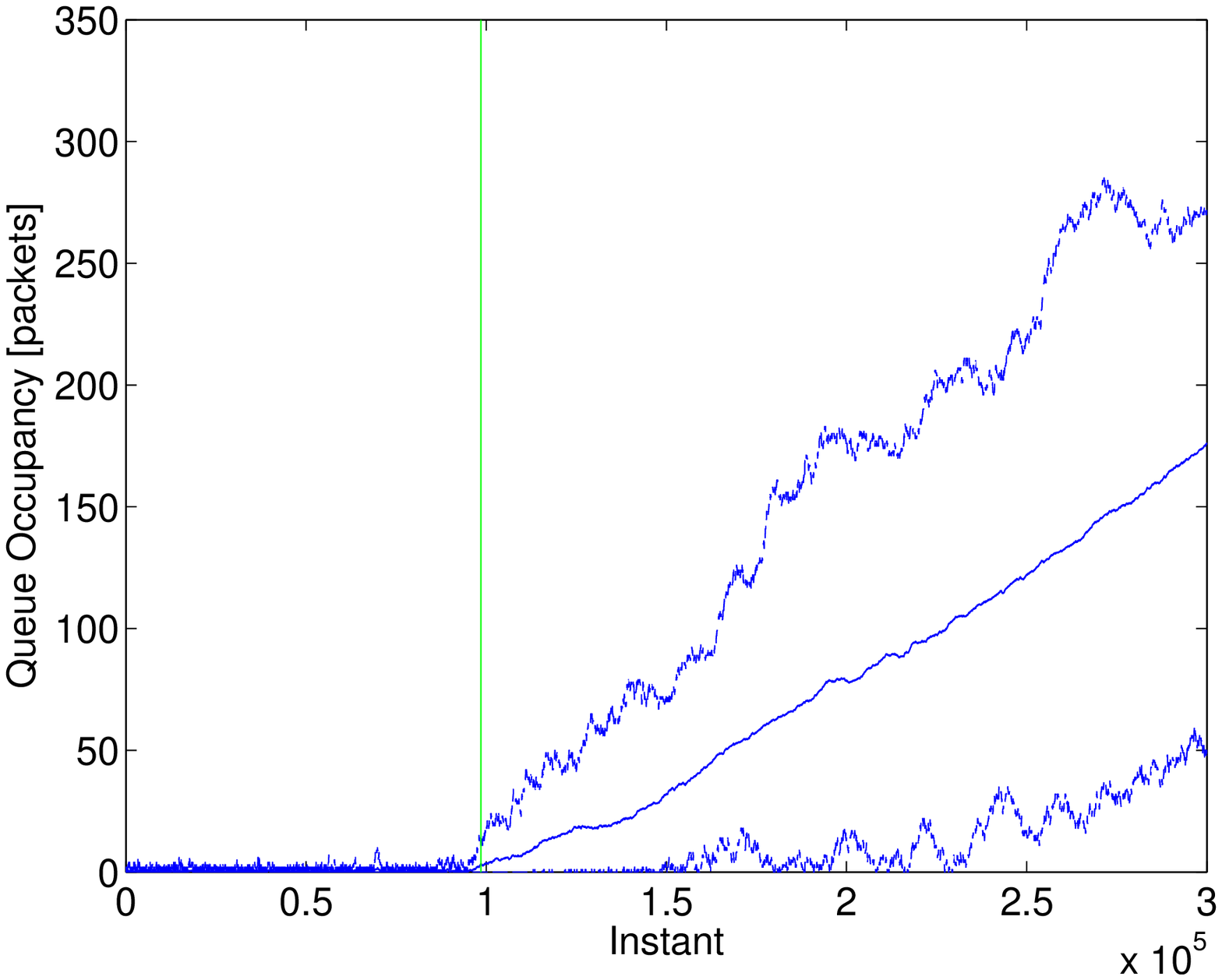}\label{fig:q_evolution_1}}\\
\subfigure[Simulation Run 2]{\includegraphics[width=2.5in]{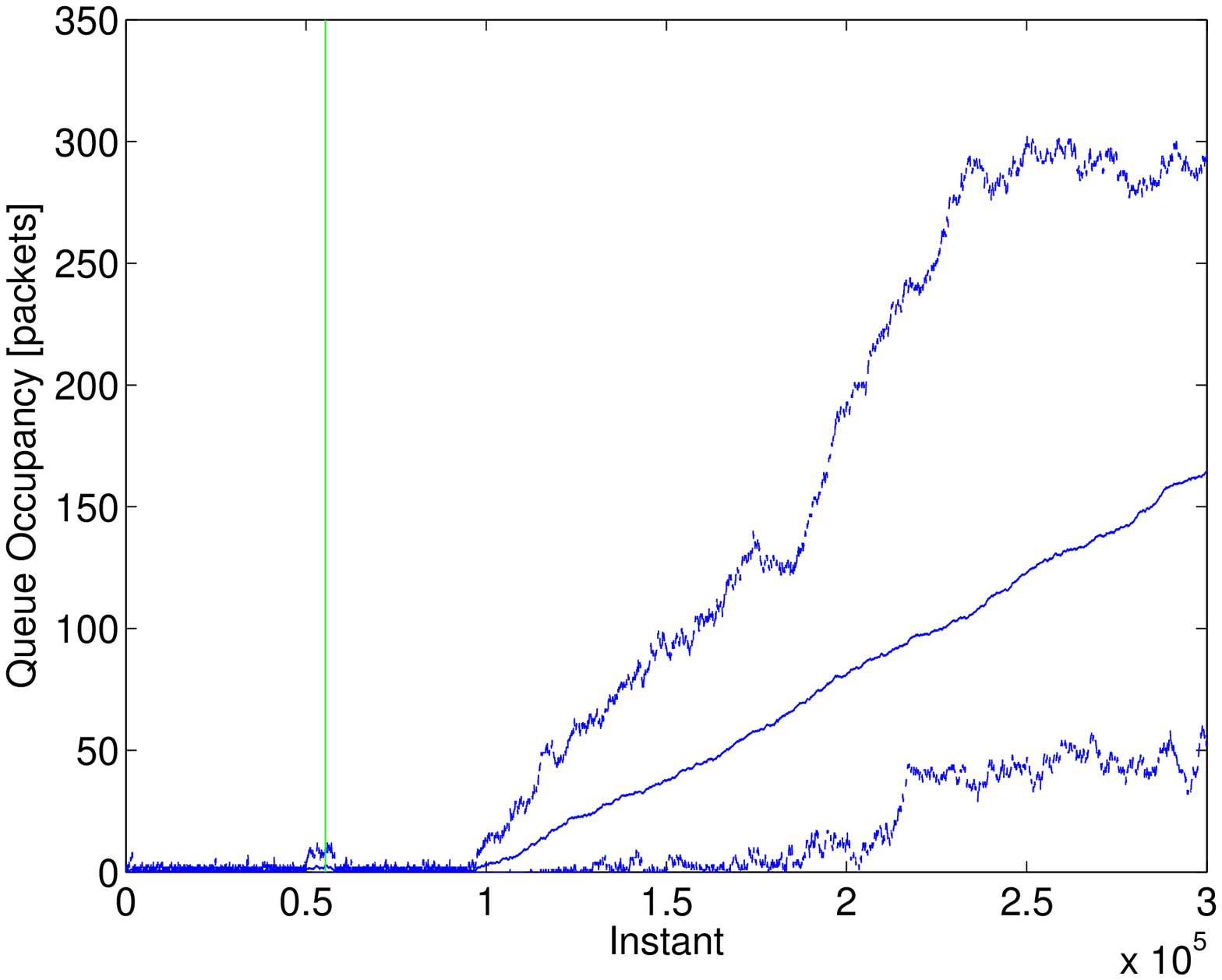}\label{fig:q_evolution_2}}
\caption{Evolution of the queue occupancy (maximum, average, minimum) of the $N$ nodes for two simulation runs in DCF ($\lambda=7.5$ packets/s). Hitting instant of first state such that $\lambda>\mu(n_x)$ also displayed (vertical line).}
\label{fig:q_evolution}
\end{figure}

Observe that in order to improve the prediction of the duration of the transitory phase, we could consider in \emph{Method 2} that after hitting a limiting state there is a certain probability to return to state $N'-1$, to remain in state $N'$ and to enter into absorption. These probabilities are dependent on the number of packets waiting for transmission at the different queues. However, recall that no information on the queue occupancies at each node is maintained in \emph{Method 2}. 

\subsubsection{Metric 2 - Last Instant $N-1$ Stations Backlogged} 

A more accurate metric is to track, after the average queue lengths have exceeded a certain threshold ($\theta$), the last instant at which $N-1$ stations were backlogged:

\begin{equation}\label{eq:TE}
T_{\rm E} = \sup \{t < T_{\rm \theta} : \exists i.x_i(t) = 0 \},
\end{equation}

where $T_{\rm \theta} = \inf\{t > 0 : {\bar{x}(t) > \theta}\}$. After $T_{\rm E}$, we can assume the transitory phase has ended as, from that instant on, $N$ stations will be contending for the channel. Thus, the network behaviour will be the stable (stationary) one given that all queues will have at least one packet pending for transmission, i.e., saturation conditions. Observe that the limitation of this metric is that it can only be assessed through experimental evaluation. Therefore, we can only evaluate this metric running Monte Carlo simulations in the system described in \emph{Method 1} and using \emph{Method 3}. However, this metric allows us to provide accurate results on the duration of the transitory phase and it can be used to asses the accuracy of \emph{Metric 1}.

\section{Results}\label{sec:results}

We present here the results obtained from the previously described evaluation methods. We first describe the outcomes obtained from a numerical analysis based on \emph{Method 2} that allows us to derive conclusions about the duration of the long transitory phase under certain conditions. Then, we perform an experimental evaluation to provide more insight into the duration and distribution of the transitory phase. The accuracy of the considered methods and metrics is also evaluated. 

\begin{figure*}[!tb] 
\centering
\subfigure[Aloha]{\includegraphics[width=2.3in]{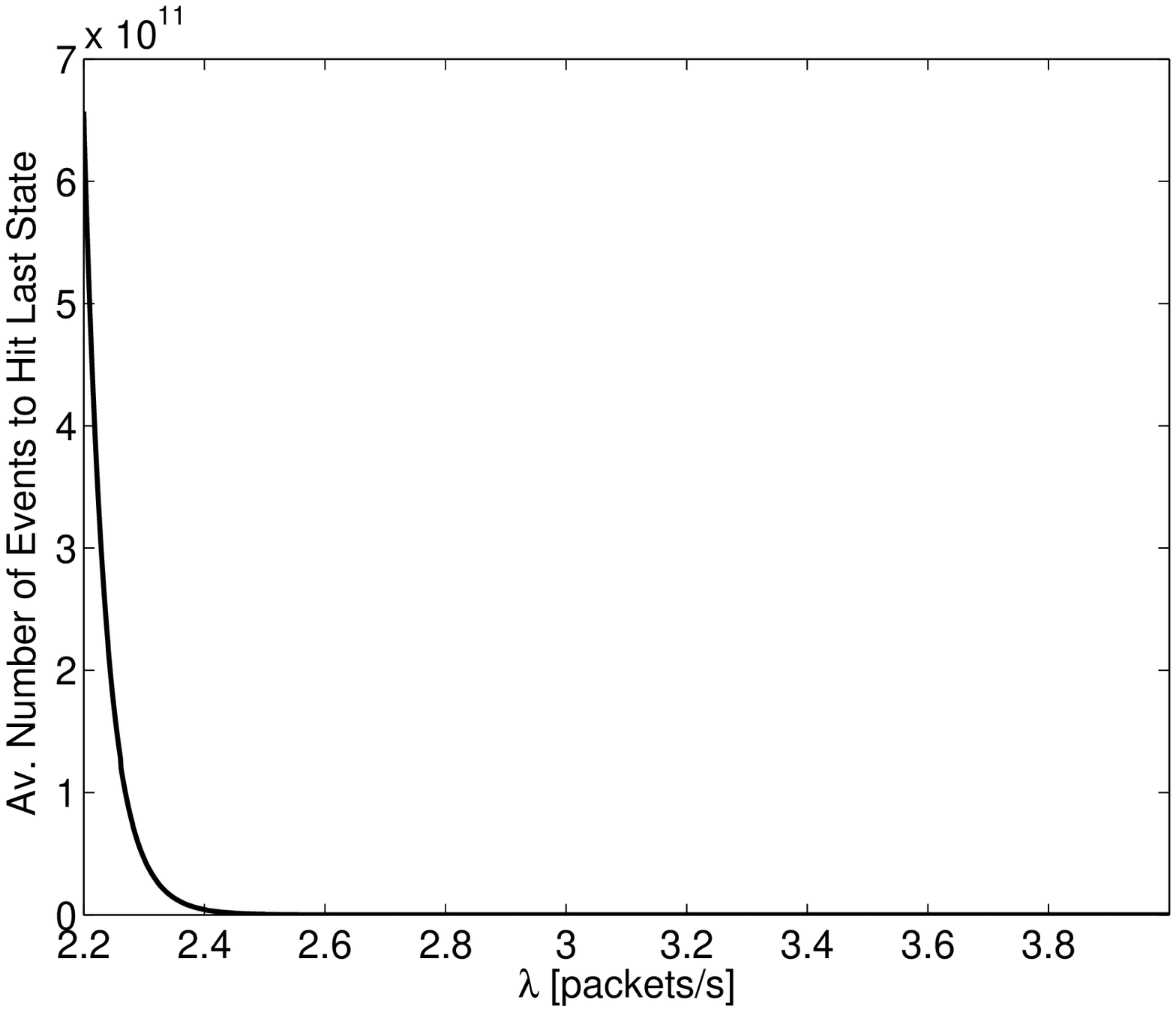}}
\subfigure[DCF ($W=8$, $m=3$)]{\includegraphics[width=2.31in]{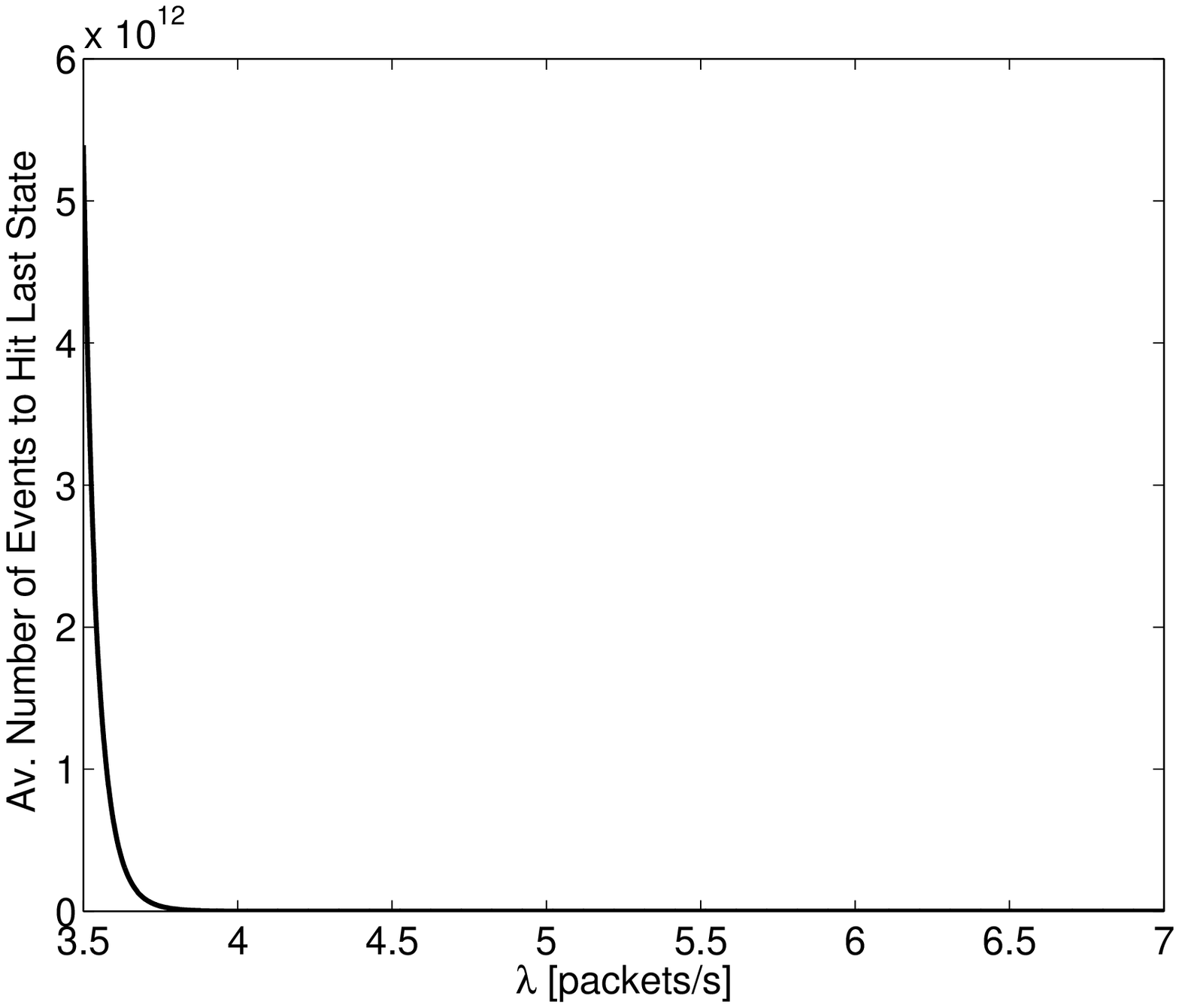}}
\subfigure[DCF ($W=32$, $m=5$)]{\includegraphics[width=2.35in]{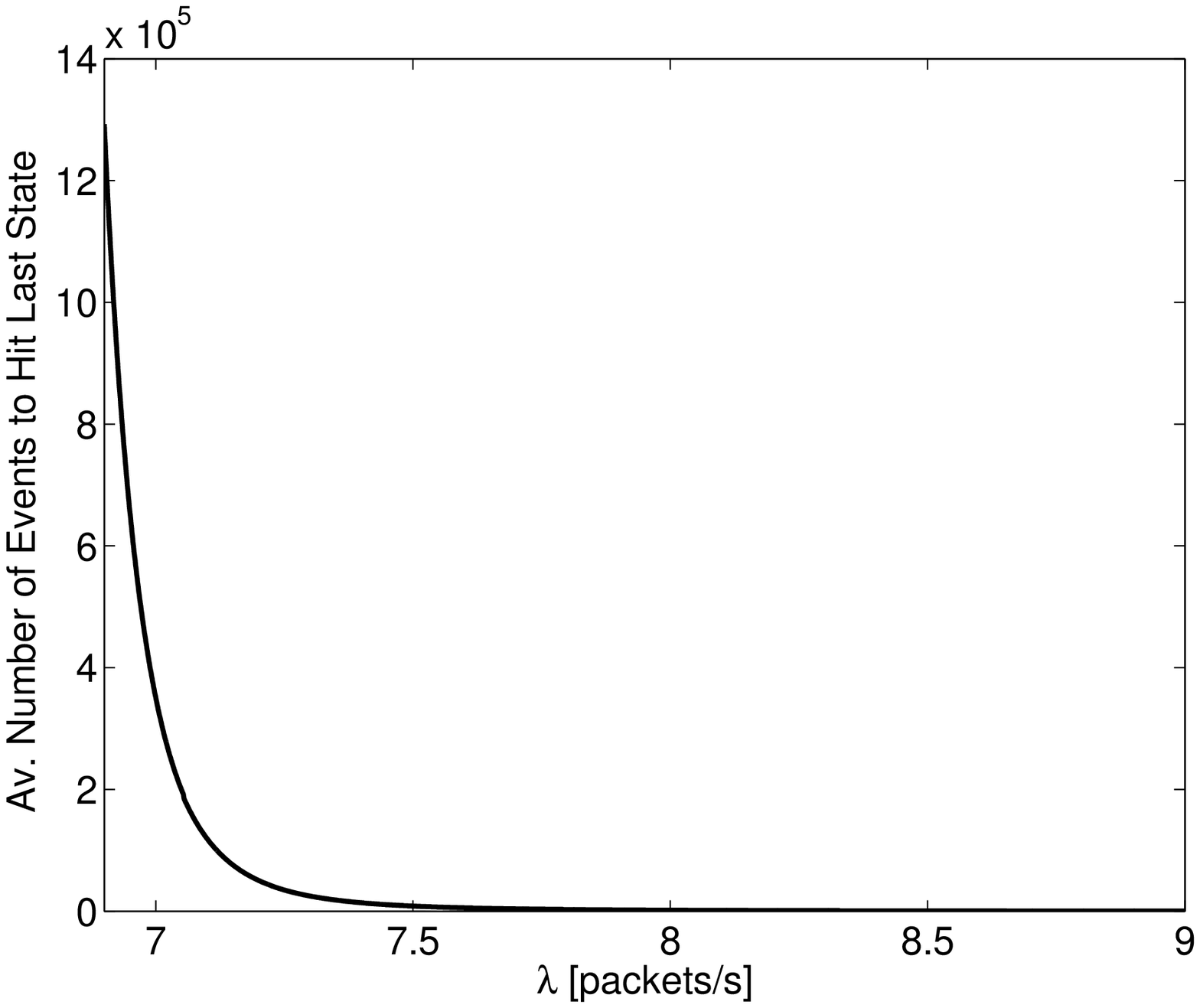}}
\caption{Results of \emph{Metric 1} (average number of events to hit state $N'$ such that $\lambda < \mu(N')$ is no longer satisfied) using \emph{Method 2} (modelling the number of backlogged stations).}
\label{fig:events_to_hit}
\end{figure*}

\subsection{Numerical Analysis}


The numerical analysis using \emph{Method 2} (modelling the number of backlogged stations) is the most efficient approach to obtain insight into the duration of the transitory phase under a range of different conditions. We present in this section the results and conclusions it provides as well as an evaluation of its accuracy by comparing results with \emph{Method 1} (the system with coupled service rates). The former allows us to demonstrate the long duration of the transitory phase and the conditions under which it occurs. The latter provides insight into the limitation of our approximation to model the number of backlogged stations instead of tracking the number of packets at each queue.

\subsubsection{Duration of the Transitory Phase} Here we show via numerical analysis using \emph{Method 2} and \emph{Metric 1} (hitting times to a limiting state) that the duration of the transitory phase is extremely long under certain conditions. We have solved the system of linear equations presented in Section \ref{sec:methods_metric1} considering $N=50$ nodes for Aloha ($W=32$) and the DCF ($W=8, m=3$ and $W=32, m=5$), which are the same parameters used in Section \ref{sec:identifying_decoupling}. The state-dependent service rates ($\mu(x)$) for the different protocols and configurations are obtained from the analytical models presented in Appendix \ref{appendix:aloha} and \ref{appendix:dcf}, respectively, considering saturated conditions. The average number of events to hit state $N'$ for different arrival rates are shown in Fig. \ref{fig:events_to_hit}. Observe that, in all cases considered, as the packet 
arrival 
rate increases, the average number of events to hit state $N'$ tends to zero. On the contrary, for reduced packet arrival rates, it is of a high order of magnitude. Considering that \emph{Metric 1} is a lower bound of the time at which the transitory phase ends, these results demonstrate that for packet arrival rates slightly higher than the maximum rate that can be achieved in saturation ($\mu(N)$), the duration of the transitory phase is extremely long.


\subsubsection{Accuracy of Method 2} Comparing results of \emph{Method 1} (the system with coupled service rates) and \emph{Method 2} (modelling the number of backlogged stations) we can get insight into the accuracy of the approximation considered in \emph{Method 2} to model the probability that a station that has transmitted a packet remains still backlogged. For this purpose we compare results of \emph{Metric 1} (hitting times to a limiting state) for both evaluation approaches. Fig. \ref{fig:accuracy_method1and2} shows the relative error of \emph{Metric 1} using both methods for the different protocols considered and various packet arrival rates. Parameters are as the ones used in Fig. \ref{fig:events_to_hit} while average values in \emph{Method 1} are obtained from $1000$ simulation runs and considering a maximum queue size equal to $1000$ packets. The magnitude of the relative error as well as how it increases with the reduction of the packet arrival rate can be observed in Fig. \ref{fig:accuracy_method1and2}. Although not shown in Fig. \ref{fig:accuracy_method1and2} (as absolute values are depicted), the value obtained for \emph{Metric 1} from \emph{Method 1} is, in all cases evaluated, higher than the one obtained from \emph{Method 2}. Therefore, despite of the large errors, the conclusions derived in the previous subsection hold as results shown in Fig. \ref{fig:events_to_hit} correspond to a lower bound of the actual duration of the transitory phase. Thus, we conclude that while \emph{Method 2} shows a long transitory phase duration, a more detailed queue model is required to estimate its actual length.

\begin{figure}[tb] 
\centering
\includegraphics[width=2.5in]{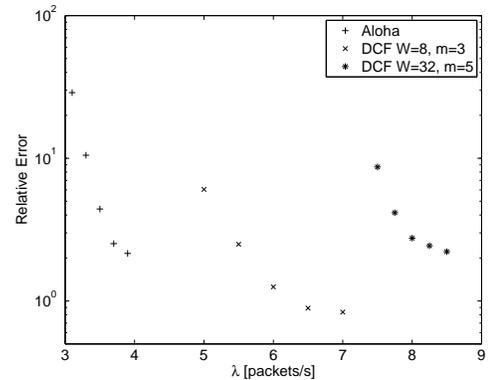} 
\caption{Relative error (ratio) of \emph{Metric1} (hitting times to a limiting state) obtained using \emph{Method 2} (modelling the number of backlogged stations) vs. \emph{Method 1} (the system with coupled service rates).}
\label{fig:accuracy_method1and2}
\end{figure}

\begin{figure*}[!tb] 
\centering
\subfigure[$\lambda=7.5$ packets/s]{\includegraphics[width=2.32in]{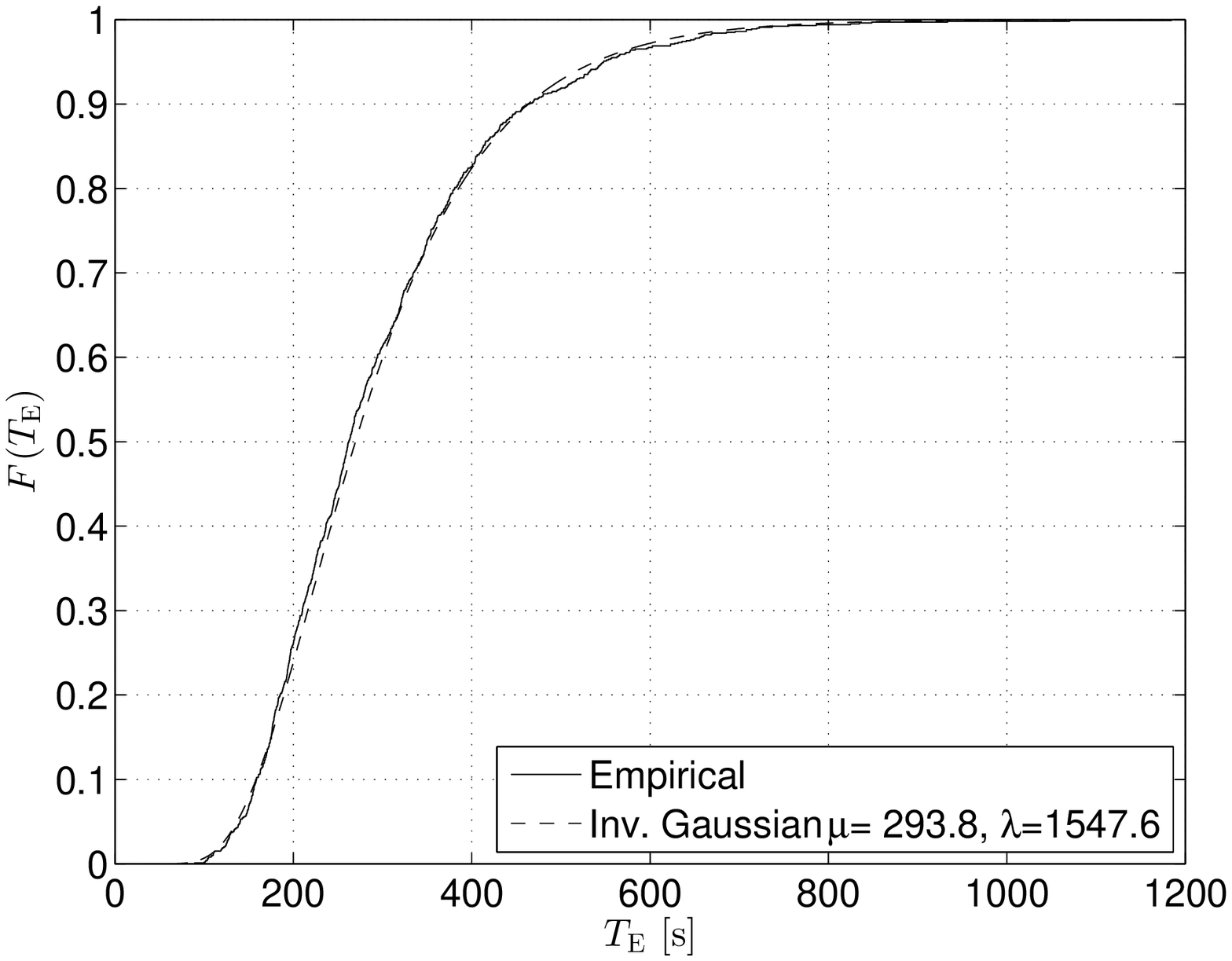}\label{fig:time_lastNminus1_hist_dcf_75}}
\subfigure[$\lambda=7.75$ packets/s]{\includegraphics[width=2.3in]{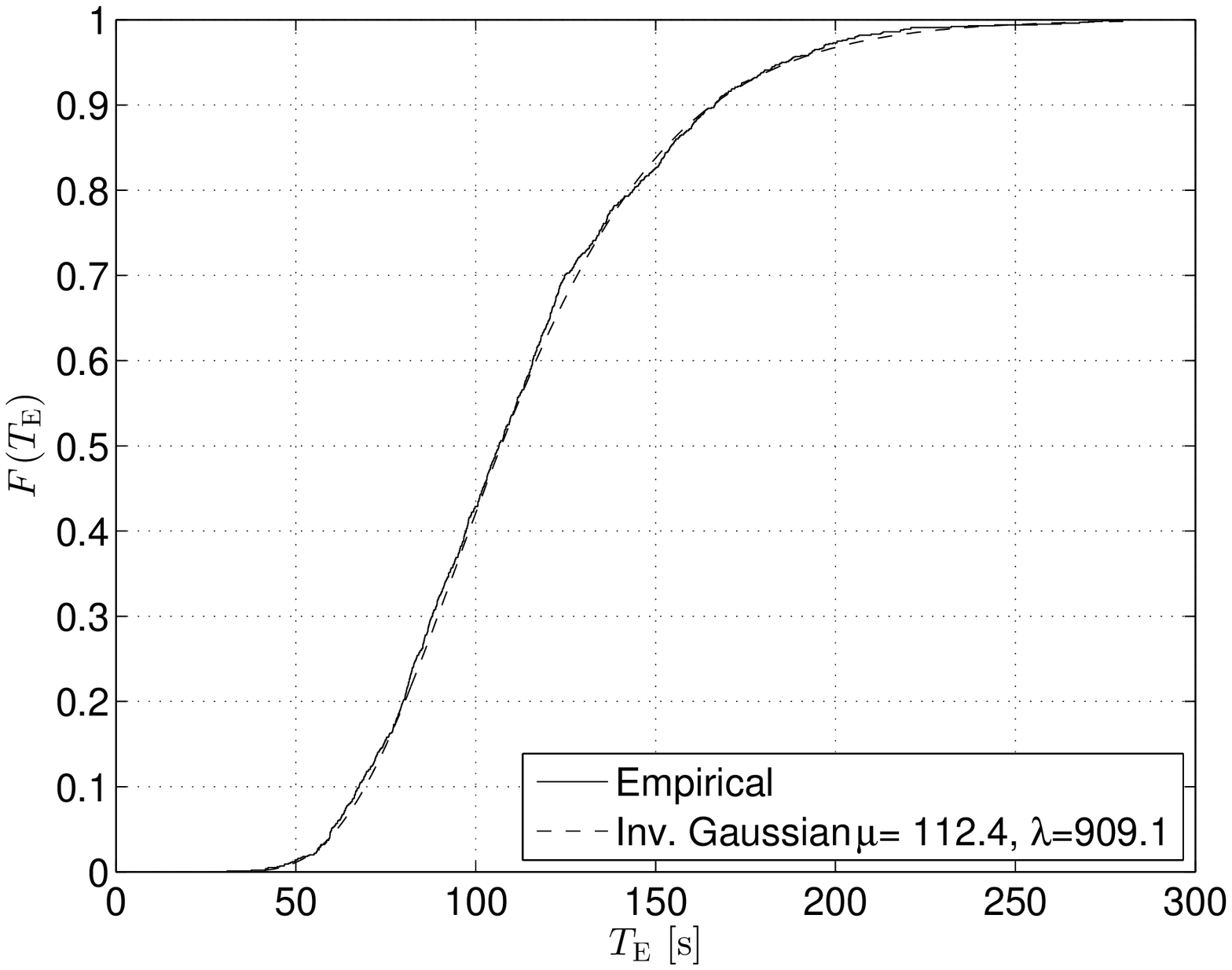}}
\subfigure[$\lambda=8$ packets/s]{\includegraphics[width=2.3in]{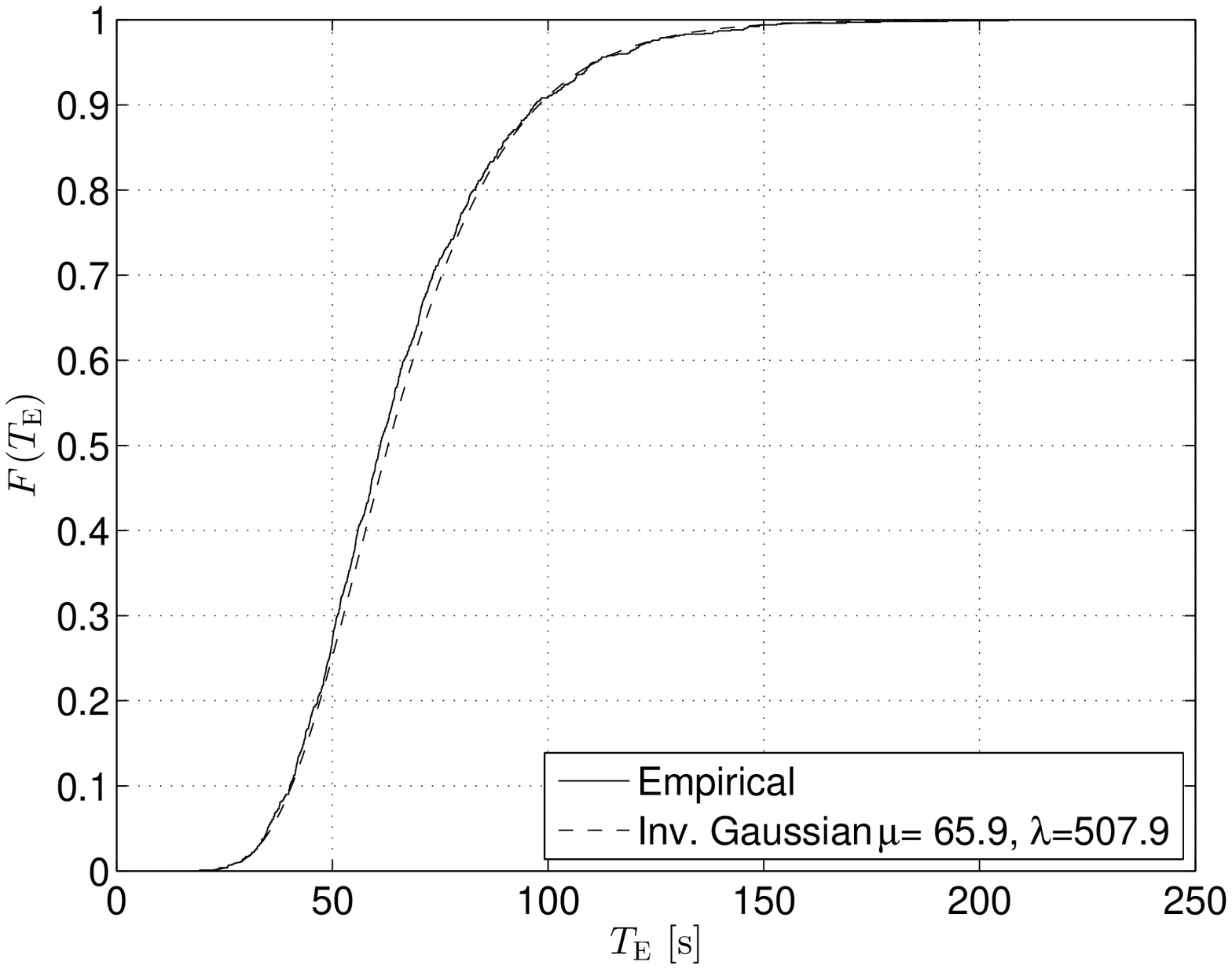}}
\caption{Empirical CDF of \emph{Metric 2} (last instant $N-1$ stations backlogged) using \emph{Method 1} (the system with coupled service rates) in DCF with $W=32, m=5$. Inverse Gaussian distribution with parameters selected to best fit the empirical distribution also depicted.}
\label{fig:time_lastNminus1_hist_dcf}
\end{figure*}

\begin{figure*}[!tb] 
\centering
\subfigure[$\lambda=7.5$ packets/s]{\includegraphics[width=2.25in]{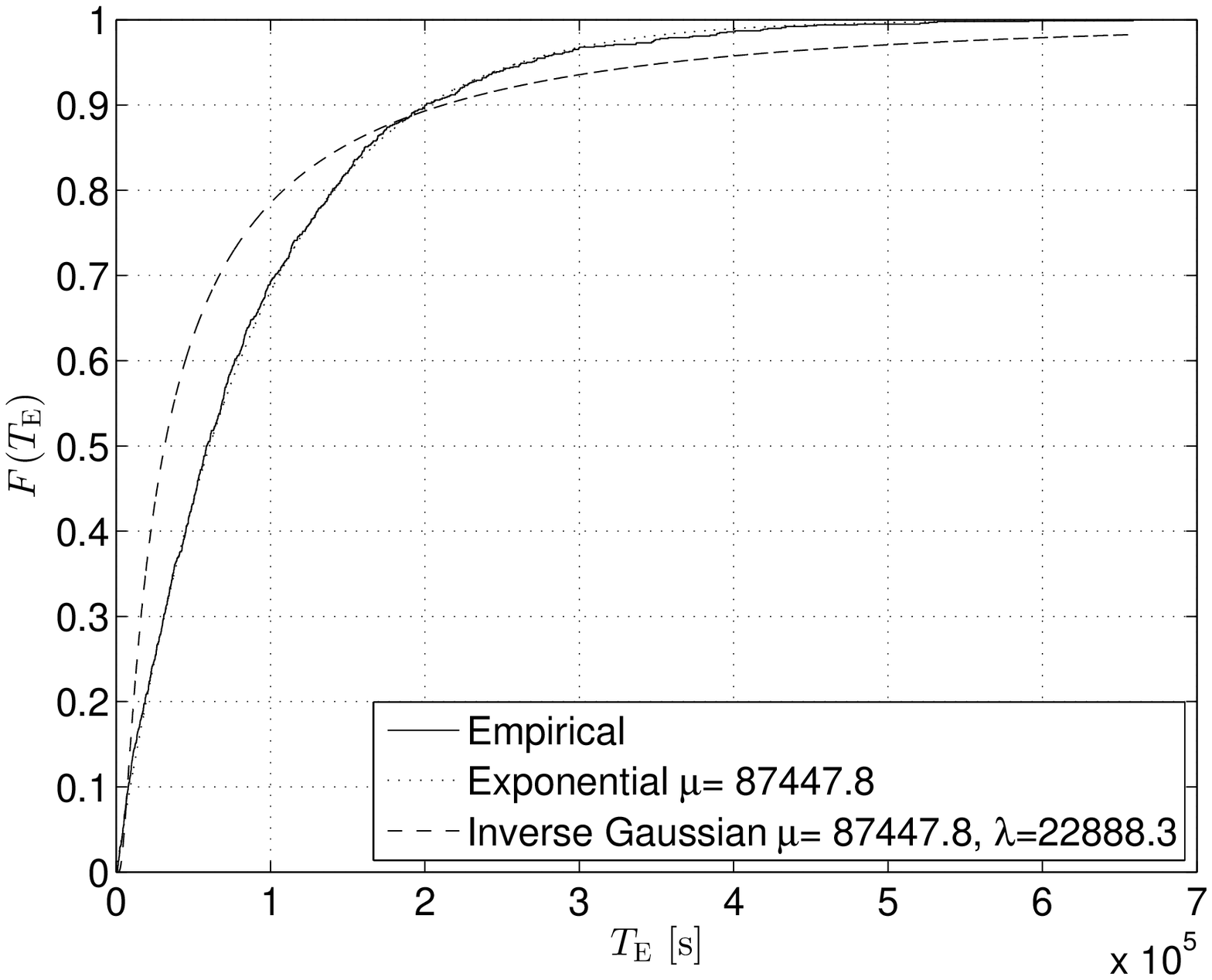}\label{fig:time_lastNminus1_hist_dcf_sense_75}} 
\subfigure[$\lambda=7.75$ packets/s]{\includegraphics[width=2.31in]{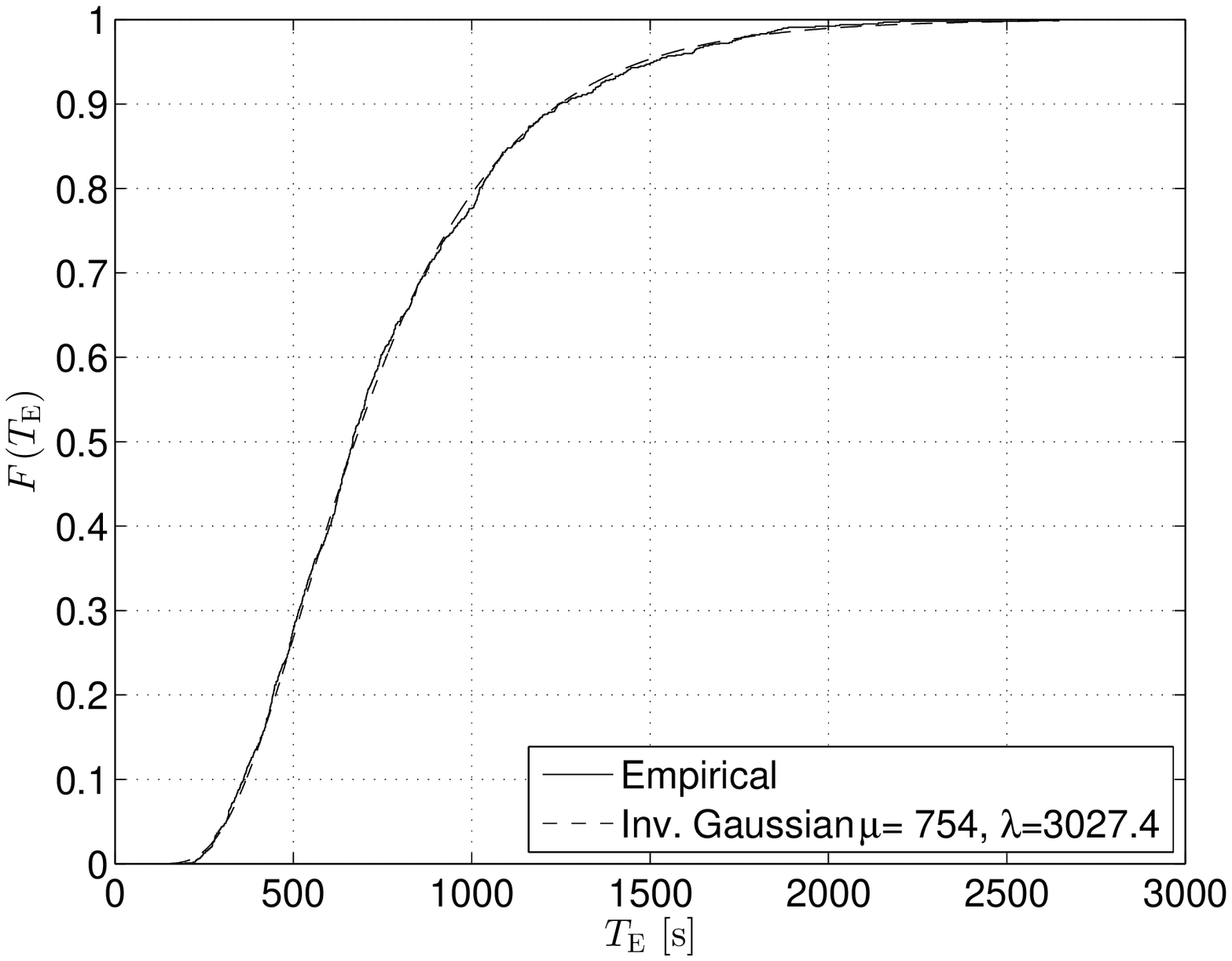}}
\subfigure[$\lambda=8$ packets/s]{\includegraphics[width=2.3in]{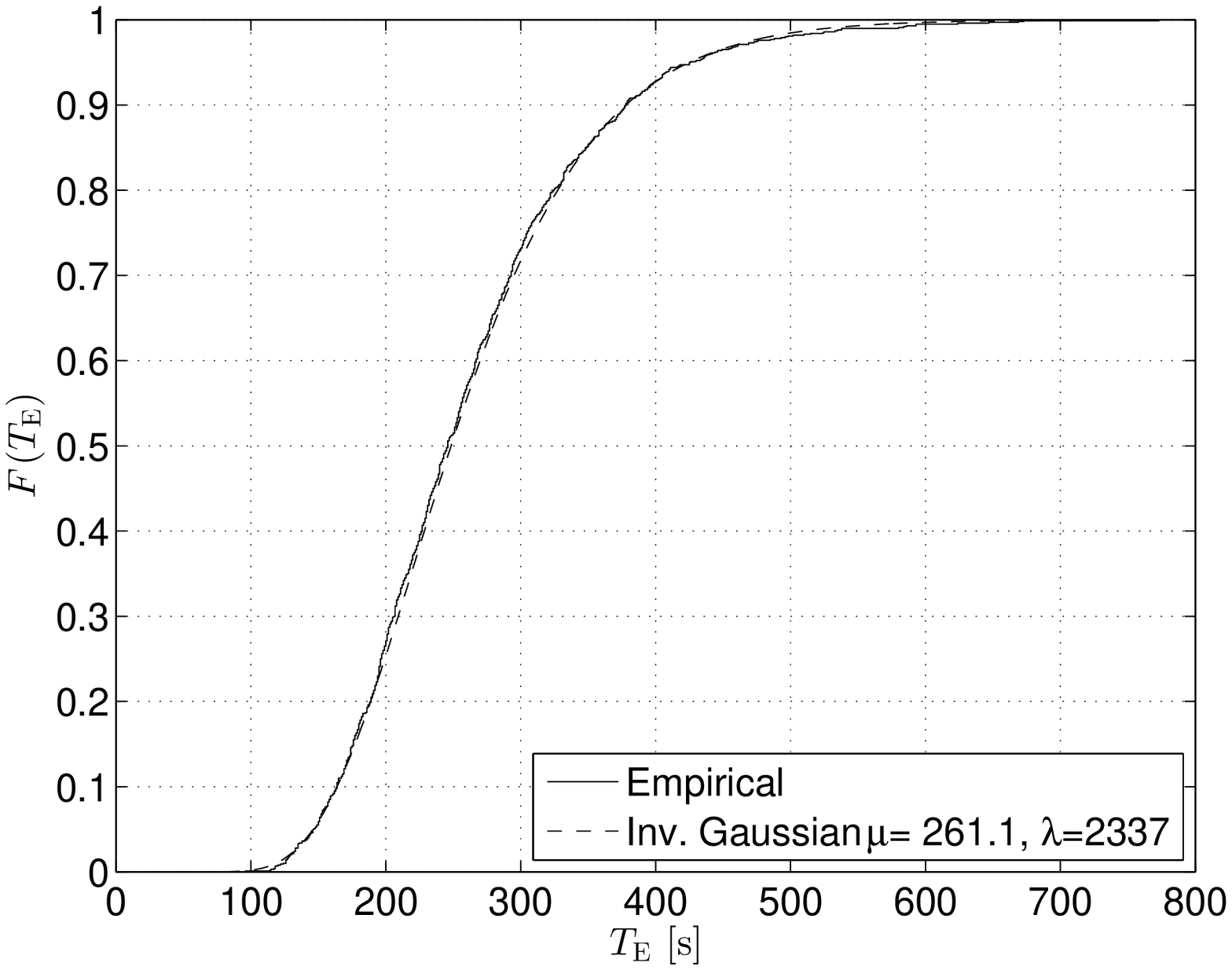}}
\caption{Empirical CDF of \emph{Metric 2} (last instant $N-1$ stations backlogged) using \emph{Method 3} (network simulator) in DCF with $W=32, m=5$. Inverse Gaussian (and exponential in Fig. \ref{fig:time_lastNminus1_hist_dcf_sense_75}) distribution with parameters selected to best fit the empirical distribution also depicted.}
\label{fig:time_lastNminus1_hist_dcf_sense}
\end{figure*}

\subsection{Experimental Evaluation}

We now perform an experimental evaluation using \emph{Method 1} (the system with coupled service rates) and \emph{Method 3} (network simulations) to obtain more insight into the duration and the distribution of the transitory phase. We also evaluate the accuracy of the analysis in \emph{Method 1} by comparing the results with those obtained in \emph{Method 3}.


\subsubsection{Duration and Distribution of the Transitory Phase}

We track the last time $N-1$ stations were backlogged once the average queue occupancy has reached a value higher than a given threshold (as described in \emph{Metric 2}). We consider this threshold to be equal to $\theta=0.75Q$ while the other considerations are as in the last subsection ($N=50$, parameters described in Section \ref{sec:identifying_decoupling} and $Q=1000$ packets). The empirical Cumulative Distribution Function (CDF) of $T_{\rm E}$ obtained from $1000$ simulation runs for different packet arrival rates for DCF with $W=32, m=5$ using \emph{Method 1} and \emph{Method 3} are shown in Fig. \ref{fig:time_lastNminus1_hist_dcf} and \ref{fig:time_lastNminus1_hist_dcf_sense}, respectively. We can observe in Fig. \ref{fig:time_lastNminus1_hist_dcf_75} and \ref{fig:time_lastNminus1_hist_dcf_sense_75} a considerable mismatch in magnitude between \emph{Method 1} and \emph{Method 3} for small packet arrival rates (we will further discuss this issue in 
the next 
subsection). However, for higher packet arrival rates, we can see that \emph{Method 1} provides 
increased accuracy. Important aspects to be noted here are the extremely long duration observed in network simulations for small packet arrival rates as well as the distribution of the metric of interest. First, note that the actual length of the transitory period in simulations when the packet arrival rate is $7.5$ packets/s is in the order of several hours (Fig. \ref{fig:time_lastNminus1_hist_dcf_sense_75}). Second, observe that for shorter durations of the transitory period (Figs. \ref{fig:time_lastNminus1_hist_dcf}(a-c) and Figs. \ref{fig:time_lastNminus1_hist_dcf_sense}(b-c)), the empirical distribution resembles that of an inverse Gaussian. In fact, both lognormal and Birnbaum-Saunders distributions provided similar goodness of fit (in terms of negative log likelihood) as the inverse Gaussian. However, we have presented the latter due to its relation to the distribution of first hitting times in ordinary diffusion processes \cite{rangarajan2000first}. Note that \emph{Metric 2} can be seen as the first 
time at which a certain boundary is reached and that the actual process can be seen as a random walk, which corresponds to a diffusion process in the scaling limit. However, for longer durations of the transitory period (Fig. \ref{fig:time_lastNminus1_hist_dcf_sense_75}), the distribution obtained can be better described as an exponential. These results suggest that the actual distribution could be described as a combination of both distributions, with the exponential one having more influence for longer durations (smaller packet arrival rates) and the inverse Gaussian being more relevant for shorter durations (higher packet arrival rates).

\subsubsection{Accuracy of Method 1}

We evaluate here in more detail the accuracy of \emph{Method 1} compared to network simulations (\emph{Method 3}). To this end, we show in Table \ref{tbl:comp_lastNminus1_dcf} the average time at which $N-1$ stations were backlogged before an average queue occupancy of $\theta=0.75Q$ is detected (i.e., \emph{Metric 2}). In order to get more insight into the accuracy of \emph{Method 1} we consider here even smaller packet arrival rates than in the last subsection (although still higher than the stability limit). As can be observed from Table \ref{tbl:comp_lastNminus1_dcf}, \emph{Method 1} provides subtantially shorter predictions of \emph{Metric 2} than what is obtained using \emph{Method 3}. Recall that the difference between \emph{Method 1} and \emph{Method 3} is the assumption of exponentially distributed service times in \emph{Method 1}. In fact, the channel access delay in DCF when considering the idle backoff periods, overhearing other transmissions as well as the effects of the collision probability, 
has been shown to be 
better 
charactised by a heavy-tailed distribution under the infinite retry limit assumption \cite{sakurai2007mac}. This suggests that the actual variability in service times is the underlining cause of obtaining longer transitory phases than those predicted considering the service time exponentially distributed. Therefore, we expect \emph{Method 1} to provide a lower bound on the time for the transitory phase to end. Note that this outcome supports the conclusions derived via numerical analysis, i.e., results shown in Fig. \ref{fig:events_to_hit} correspond to an effective lower bound of the actual time the system moves to the stable behaviour.

{\renewcommand{\arraystretch}{1.3} 
\begin{table}[tb!]
\centering
\caption{Comparison of \emph{Metric 2} (last instant $N-1$ stations backlogged) obtained using \emph{Method 1} (the system with coupled service rates) and \emph{Method 3} (network simulator) in DCF with $W=32, m=5$.}\label{tbl:comp_lastNminus1_dcf}
\begin{tabular}{|c|c|c|} \hline
$\lambda$ [packets/s] & $T_{\rm E}$ (\emph{Method 1}) & $T_{\rm E}$ (\emph{Method 3}) \\ \hline
$7.1$ & 31.76 h & - \\ \hline
$7.2$ & 2.05 h & - \\ \hline
$7.25$ & 51.17 min & - \\ \hline
$7.5$ & 4.90 min. & 19.57 h\\ \hline
$7.75$ & 1.87 min. & 12.57 min. \\ \hline
$8$ & 1.10 min. & 4.35 min. \\ \hline
\end{tabular}
\end{table}

\section{Practical Implications}\label{sec:implications}

In this section, we discuss different implications of the long transitory phase for real implementations. First, we describe the implications of the transitory being of extremely long duration, then we overview the effect of the assumption of exponential interarrival of packets and finally, we provide a practical way to take advantage of the high-throughput phase.


\subsection{The Extremely Long Duration}

We have shown that, under certain conditions, the system rapidly moves to the stable operation. However, there is also the potential to face a long transitory phase. In the latter case, the network has to remain under the same conditions for a long period of time (of the order of magnitude of hours). However, this situation is unlikely to occur in current wireless networks in which moderate dynamics, such as nodes joining/leaving the network and traffic pattern variations, may be present. Thus, under certain conditions (especially at packet arrival rates slightly higher than the maximum service rate), results from the transitory-phase may actually be the performance obtained in real scenarios. Therefore, it is important to consider both results when analysing the performance of these networks. In the same way that the results from the transitory phase only describe a certain behaviour of the network, to only consider stable performance can also produce disagreement between the actual and the predicted 
evaluation, especially in scenarios with high network dynamics.

\subsection{Non-Poisson Traffic}


One of the crucial assumptions in this work is the consideration of Poisson traffic. This is a common assumption in many analytical models and experimental evaluations. However, it is important to highlight here that this consideration has a direct impact on the probability of having a certain number of stations with a packet pending for transmission. Moreover, it also affects the probability of having a node leave the set of contending stations once a packet is transmitted, i.e., the probability that after a transmission the station is left with an empty queue. In practical scenarios, some packet interrarival times may be well characterised with an exponential distribution. However, other traffic sources as constant bit rate (coming from multimedia applications) or bursty traffic may also be present. Moreover, a combination of all these may occur at the same or at different stations in the network. 

When considering only constant bit rate traffic, we have observed in simulations that the network either remains in the transitory phase for the whole experiment or immediately moves to the stable behaviour. In the former case, there seems to be an artificial scheduling caused by the traffic source, in which the probability that several stations have a packet pending for transmission at the same time can be considered negligible. The latter case is observed when the instant at which the first packet is generated coincides in time (or it is very close) among the different contending stations. In that case, the transitory phase does not occur and the network operates in the stable phase from the start-up. 

A similar behaviour is expected to be found with bursty traffic. Note that, depending on the burst size and interval between bursts, stations will either not coincide with packets pending for transmission or they will simultaneously contend at certain intervals. These situations can lead to an oscillating behaviour in which saturation throughput may be observed in certain periods, while in others, the network may be operating in the high-throughput phase. Thus, we believe that the traffic arrivals will influence what combination of transitory and long term behaviour dominate in practice.

\begin{figure*}[!tb] 
\centering
\subfigure[$\lambda=7.5$ packets/s]{\includegraphics[width=2.3in]{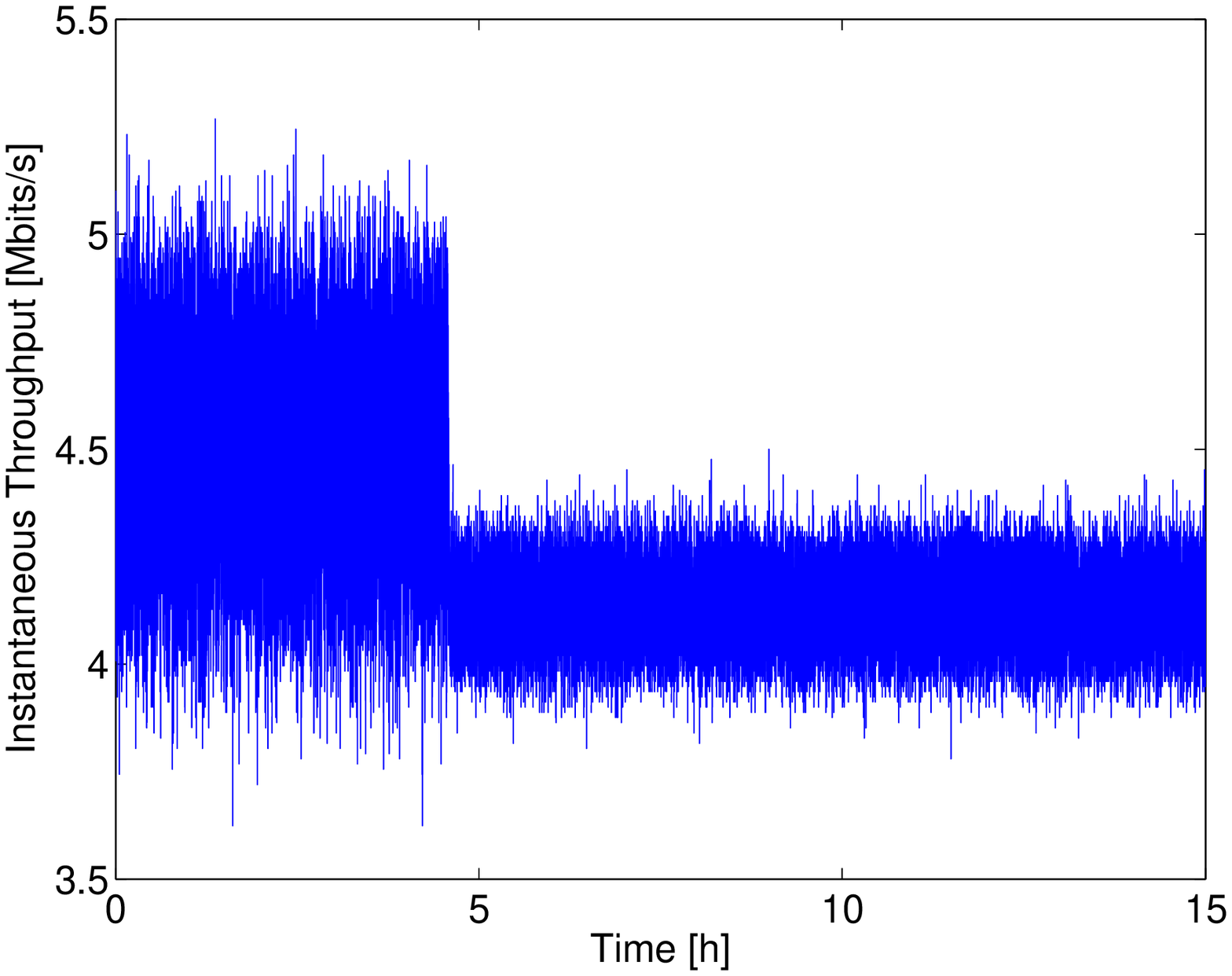}}
\subfigure[$\lambda=7.75$ packets/s]{\includegraphics[width=2.3in]{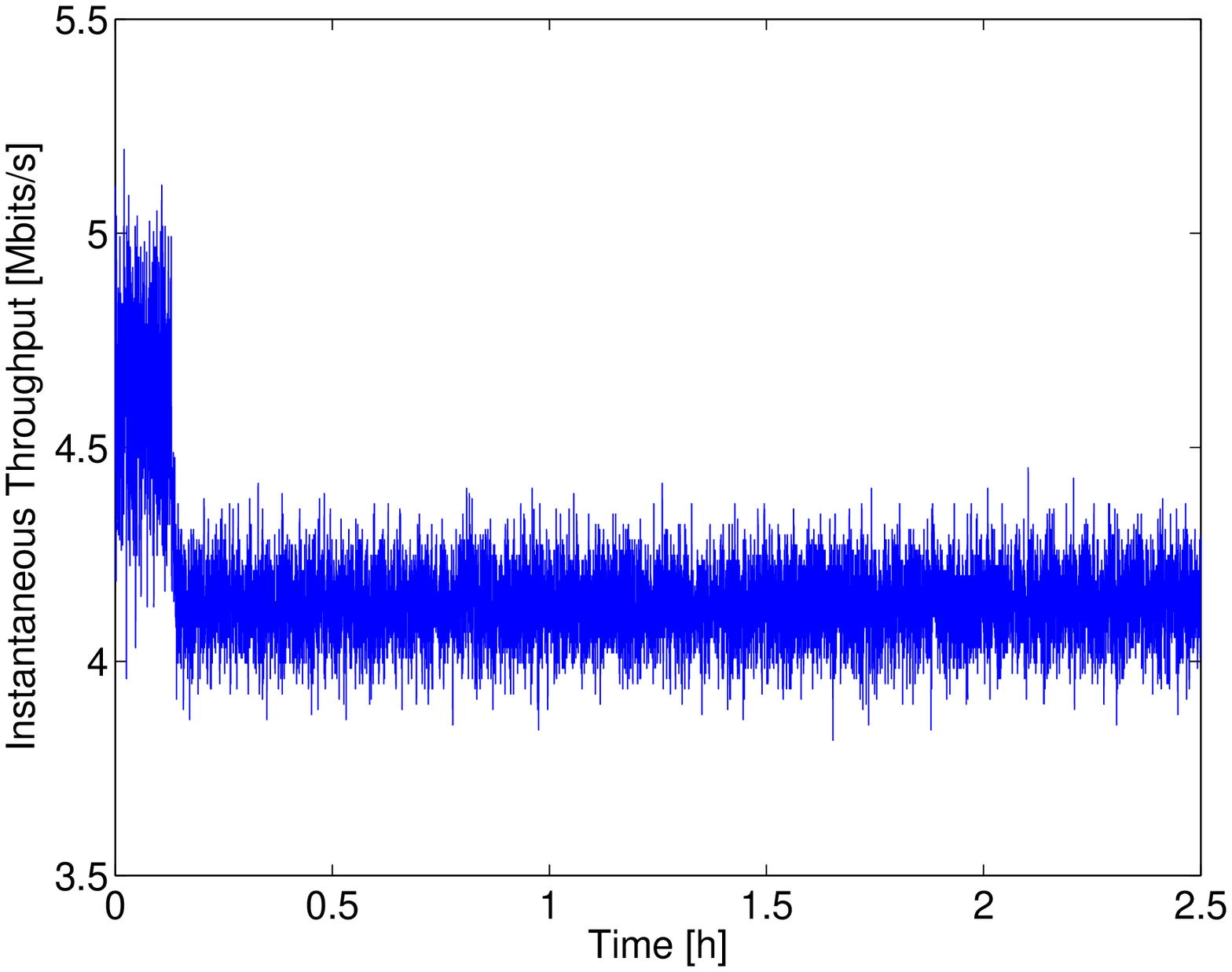}}
\subfigure[$\lambda=8$ packets/s]{\includegraphics[width=2.25in]{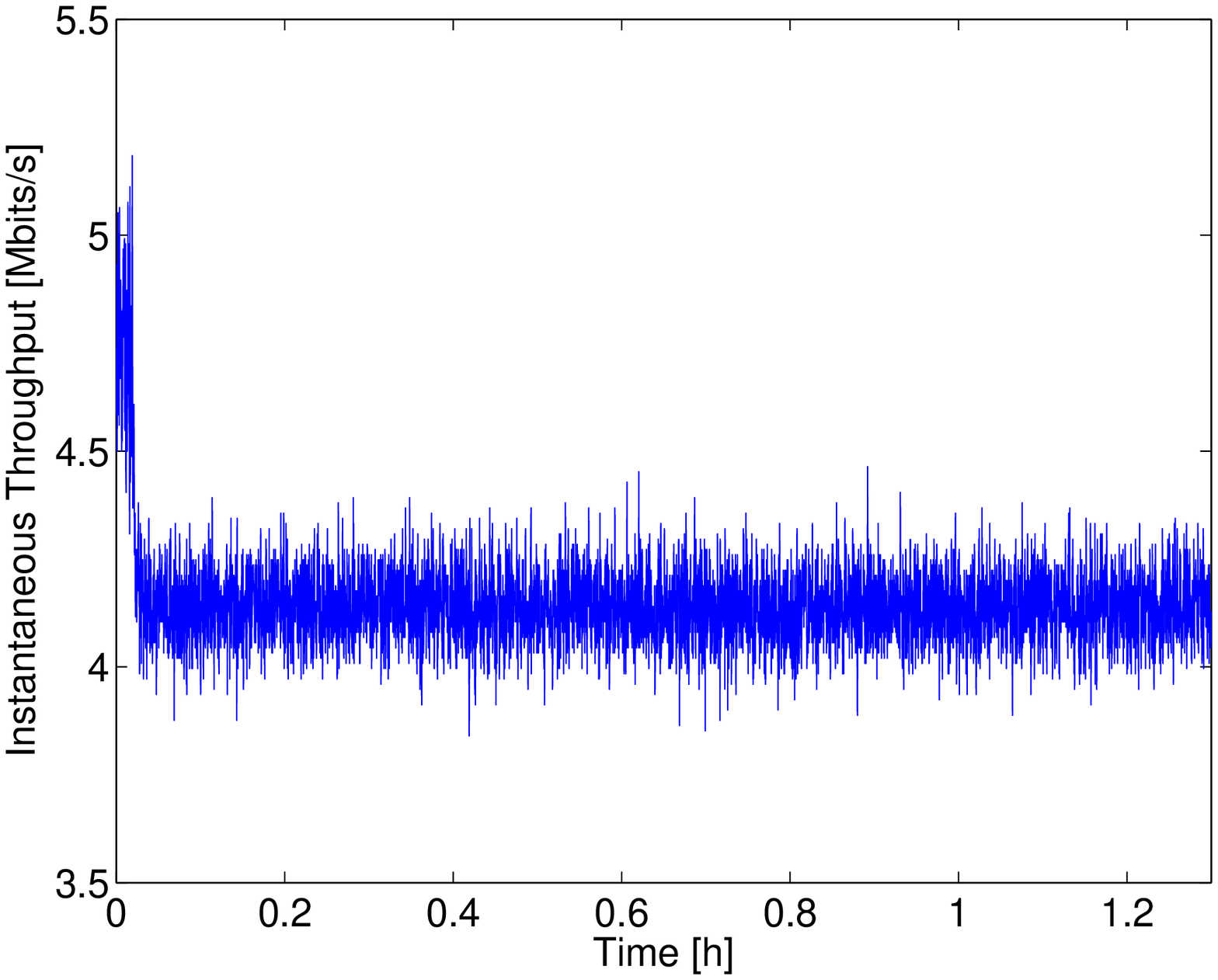}}\\
\caption{Evolution of instantaneous throughput from a single simulation run using \emph{Method 3} (network simulator) in DCF with $W=32, m=5$.}
\label{fig:temporal}
\end{figure*}

\begin{figure*}[!tb] 
\centering
\subfigure[$\lambda=7.5$ packets/s]{\includegraphics[width=2.3in]{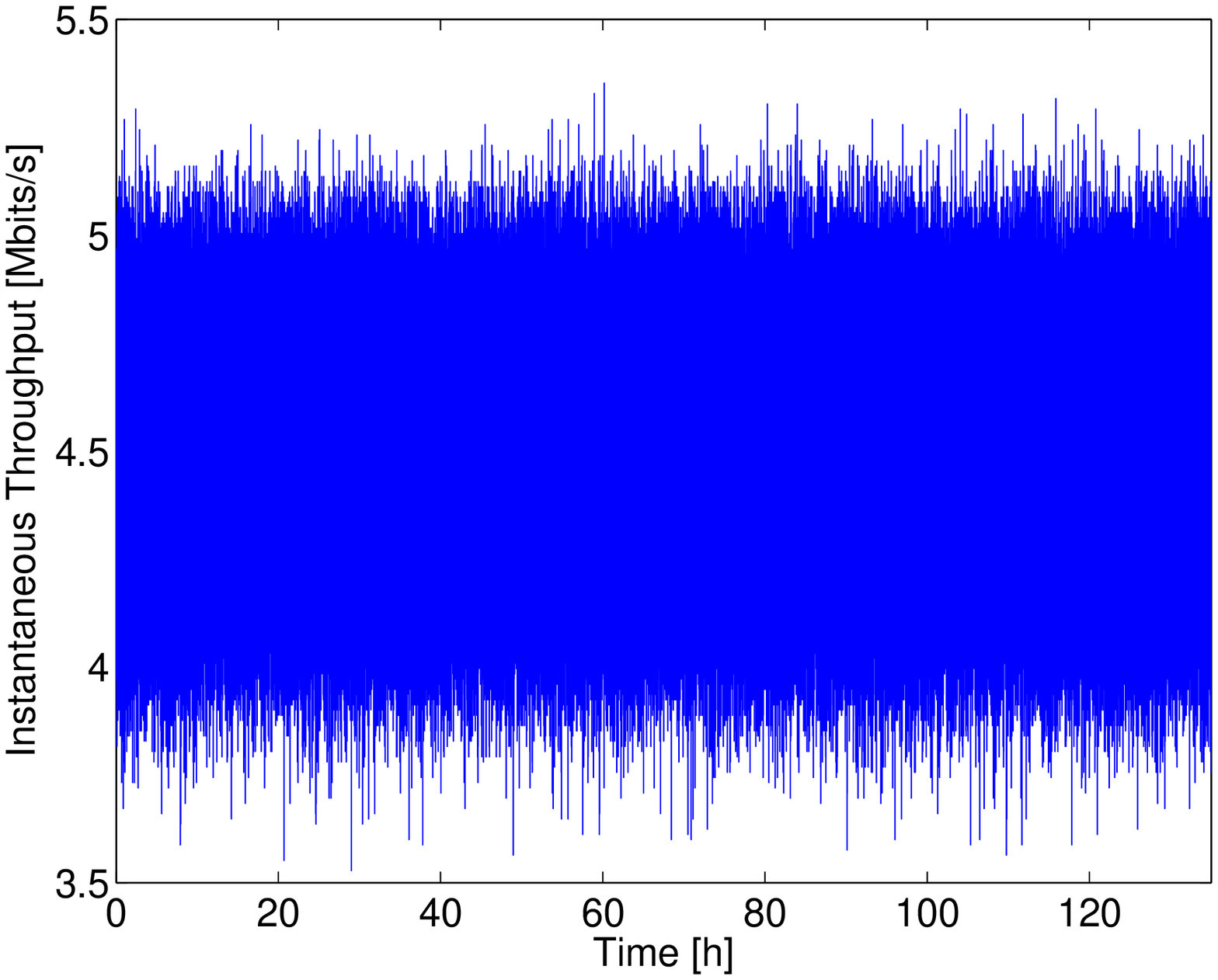}}
\subfigure[$\lambda=7.75$ packets/s]{\includegraphics[width=2.3in]{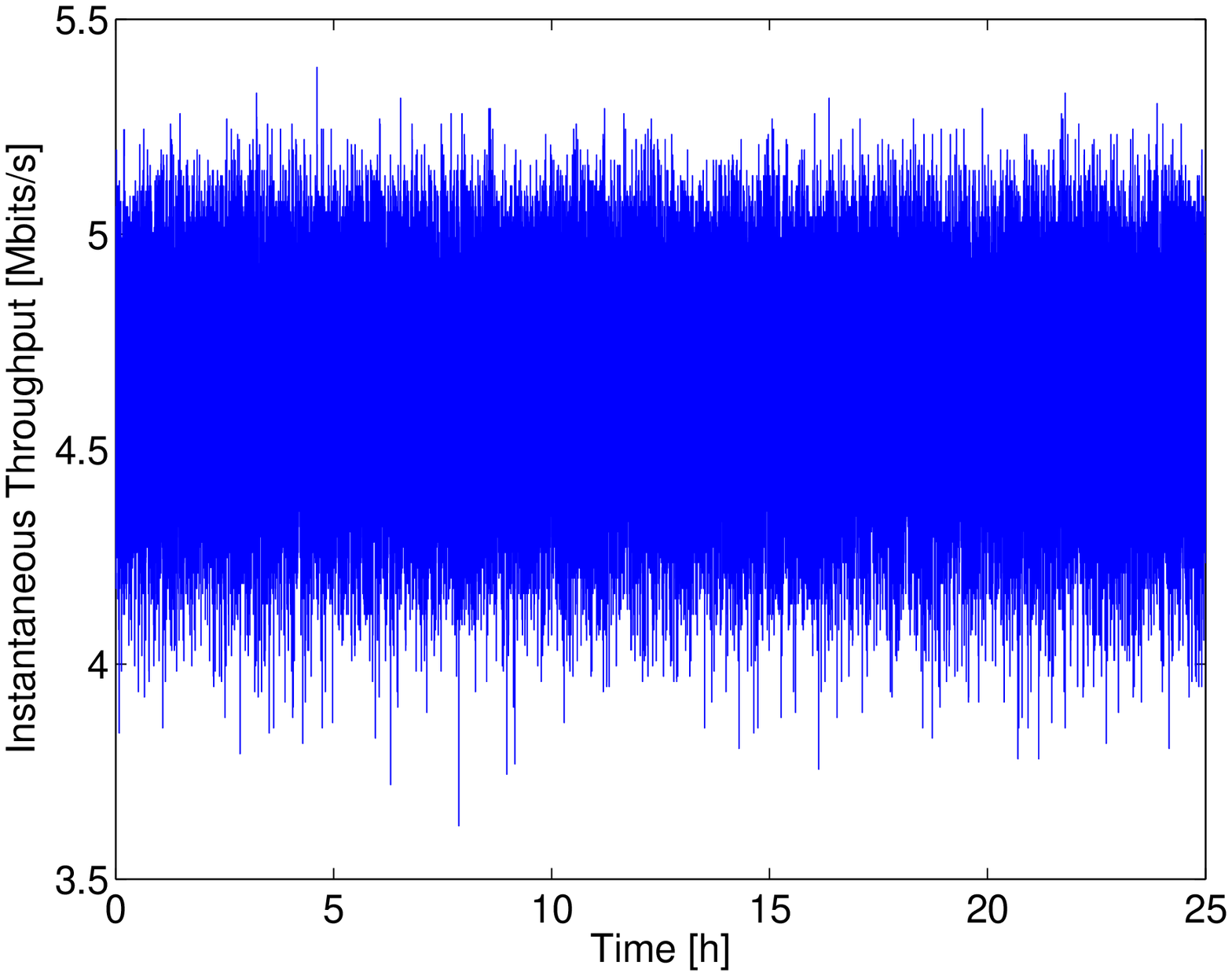}}
\subfigure[$\lambda=8$ packets/s]{\includegraphics[width=2.3in]{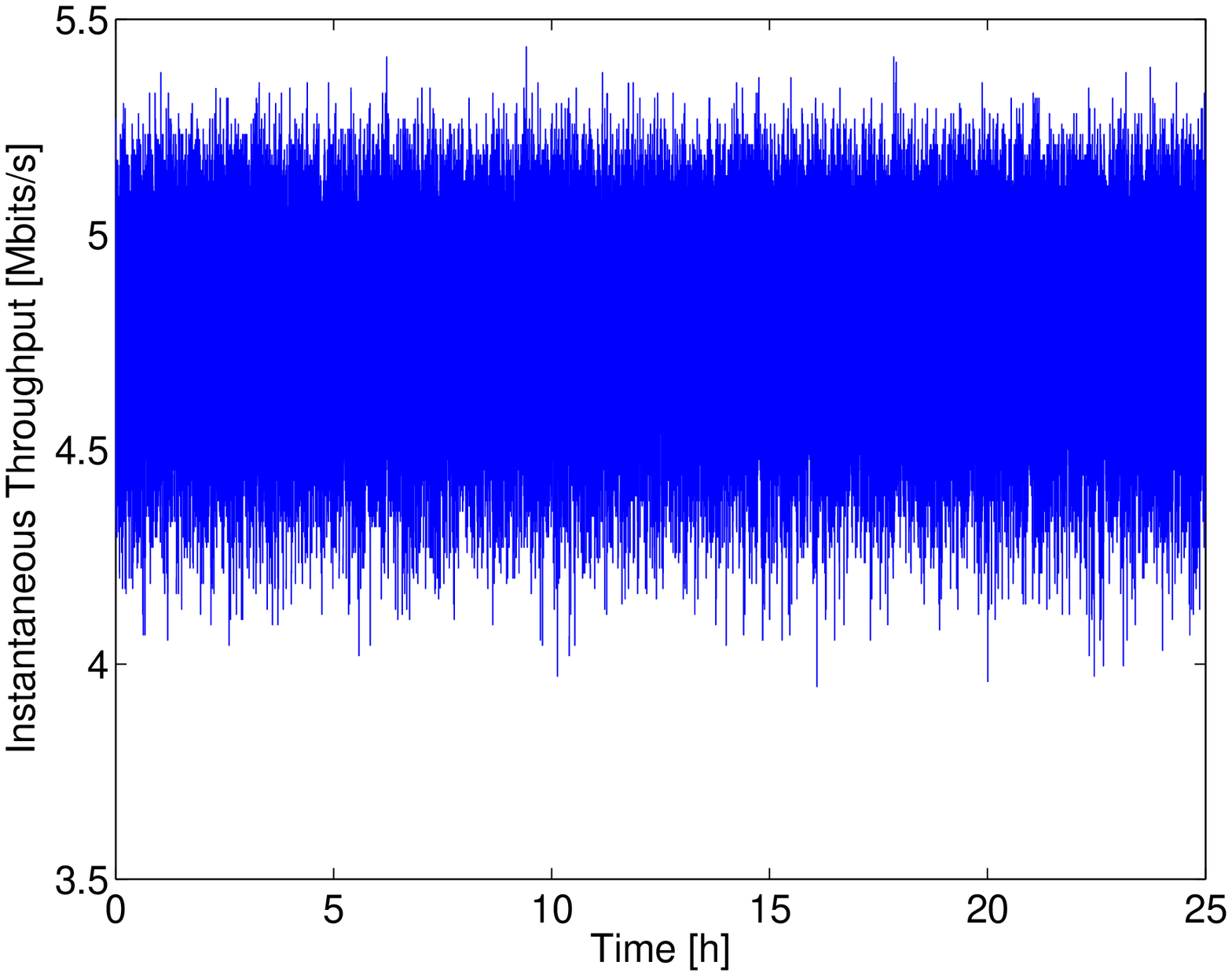}}\\
\caption{Evolution of instantaneous throughput from a single simulation run using \emph{Method 3} (network simulator) in DCF with $W=32, m=5$ with exponential delay after successful packet transmissions.}
\label{fig:temporal_delay}
\end{figure*}

\subsection{Exploiting the High-Throughput Phase}

Since the transitory phase corresponds to high-throughput performance, it can be beneficial to take advantage of the existence of the long transitory phase. Observe that, given that there is a small collision probability in the transitory phase, it is desirable that stations operate under these conditions as long as possible. One way to force the system to move back to the transitory behaviour is to delete all packets in the queues when saturated behaviour is detected, as proposed in \cite{suleiman2008impact}. Assuming that the packets loaded in the queues experience an unacceptable large delay in saturation conditions, it can be beneficial to delete them and rely on higher layers for retransmission. However, it has to be considered that a non-coordinated reset of the buffer will not ensure the return to the transitory period. Thus, a centralised controller that detects and notifies stations to take this action simultaneously is needed. Furthermore, an increase of traffic load may appear due to retransmissions at 
higher layers, creating a potential for 
oscillatory behaviour. 

Here we propose a distributed and simple solution to further increase the duration of the long transitory phase. To this aim, we introduce an exponentially-distributed delay after successful transmissions at the MAC layer during which stations do not attempt transmission of the next packet in the queue (if any). Note that this mechanism is equivalent to artificially increase the probability that a station that has successfully transmitted is left with an empty buffer. With a proper setting of this artificial delay, we can reduce the probability of having a high number of stations contending for the channel and thus, keep the network in the transitory phase. To illustrate the benefits of this proposal we show in Fig. \ref{fig:temporal} the temporal evolution of the throughput (in $1$s intervals) obtained from \emph{Method 3} for different packet arrival rates using DCF with $N=50, W=32, m=5$ and the same parameters used through this article (shown in Table \ref{tbl:parameters}). The change in behaviour can be 
clearly observed: the 
instantaneous 
throughput moves from one of the solutions predicted by the analytical model (solution labelled with \emph{Analysis 1} in Fig. \ref{fig:renewal_reward_dcf_tp_cw32_m5}) to the saturated solution (labelled as \emph{Analysis 2} in the same figure) at a given instant of time. Then, in Fig. \ref{fig:temporal_delay}, we show the instantaneous throughput achieved by adding the extra delay after successful packet transmissions (exponentially distributed with mean equal to $2/\mu(N)$). We can observe how there is no change of behaviour and the network operates in the high-throughput phase (solution labelled as \emph{Analysis 1} in Fig. \ref{fig:renewal_reward_dcf_tp_cw32_m5}) for a longer time than the maximum transitory duration observed in our experimental evaluation (Fig. \ref{fig:time_lastNminus1_hist_dcf_sense}). Note that operating in non-saturated conditions results in a smaller average delay compared to when the network is saturated. Thus, in general, despite adding an extra delay per successful 
packet transmission, the delay performance of the network is improved. Furthermore, this technique may prove useful to prevent oscillatory behaviour caused by traffic patterns other than Poisson distributed.

\section{Final Remarks}\label{sec:conclusions}

In this work we have first demonstrated that there may be a potentially long transitory phase in many random access protocols when we operate right after the stability limit and under certain circumstances, such as infinite (or large enough to be considered infinite) buffer size and exponentially distributed interarrival of packets at the MAC layer. For this purpose, we have defined a simplified analytical model that considers the number of backlogged stations instead of keeping track of the queue occupancies at each node. This approximation has allowed us to compute an effective lower bound on the actual duration of the transitory phase by making the analysis tractable and amenable to numerical evaluation.

Second, with the goal of providing more insight on the duration of the transitory phase, we have performed an experimental evaluation using both \emph{i)} Monte Carlo simulations of a system of parallel coupled queues and \emph{ii)} a network simulator. Experimental results validate the analytical findings and show the duration of the transitory phase to be in the order of hours under certain conditions and well characterised by a combination of inverse Gaussian and exponential distributions.

We have also discussed the practical implications of our findings. On one side, we state that a complete evaluation under the circumstances described in this work must consider both, the transitory as well as the stable results. Given the extremely long transitory duration, the change in behaviour may be difficult to observe in real implementations where high dynamics are present. Thus, performance results from the transitory, instead of the stable, phase may correspond to observations in real deployments. We have also highlighted the importance of the assumption of exponential interarrival of packets. In practical implementations, we may find traffic patterns that differ from this consideration, which may affect the duration of the transitory phase as well as potentially causing an oscillatory behaviour in network performance. Finally, we have also suggested a distributed and simple method to keep the network operating in the transitory phase. Given the substantial difference in throughput that can be 
observed in certain cases between the transitory and stable phases, maintaining the network in the transitory period can provide substantial benefits in throughput as well as delay.   

We have established the relation between performance misprediction errors with \emph{i)} the use of iterative solvers of analytical models based on the decoupling approximation and \emph{ii)} the presence of the extremely long transitory phase in experimental evaluations. We believe these findings, along with the characterisation of the duration of the transitory phase we have also presented in this work, are necessary to draw complete conclusions on network performance. Moreover, given the significant potential magnitude of these misprediction errors as well as their impact on predicted performance, optimisation analysis and MAC parameter configuration, we consider the outcomes of this work even more relevant.

\appendices
\section{Renewal Reward Analysis - Aloha}\label{appendix:aloha}

We take a renewal reward approach \cite{kumar05,bianchi05,medepalli2005system,BBellalta-Eurocon2005} motivated by the fact that the attempt rate of a given node can be viewed as a regenerative process. This approach allows us to compute metrics of interest without the need to solve all state probabilities of the Markov Chain embedded in the analysis. We also apply the decoupling approximation to model both: \emph{i)} the conditional (given that a packet is transmitted) collision probability and \emph{ii)} the buffer occupancy probability right after a transmission, as independent of the backoff stage at which the packet is transmitted.

The rest of assumptions and considerations taken into account are the typical: \emph{i)} infinite, or large enough to be considered infinite, buffer size and retry limit, \emph{ii)} exponentially distributed interarrival of packets, \emph{iii)} ideal channel conditions, and \emph{iv)} that all nodes are in mutual coverage range, that is, all nodes can overhear each other's transmissions.

Assuming an infinite buffer size, the mean queue occupancy ($\rho$) of a node is derived considering the time needed to release a packet from the queue ($D$), called service time, and the packet arrival rate from the network layer ($\lambda$) as:

\begin{equation}\label{eq:rho}
\rho=\min(\lambda D,1).
\end{equation}

The service time is computed as the sum of the following three components: \emph{i)} the total time on average spent in transmitting packets that result in a collision, \emph{ii)} the time spent successfully transmitting the packet and \emph{iii)} the total average backoff duration until the successful frame transmission (equal to $\frac{W}{2}\sigma$ for Aloha, $\sigma$ being the duration of an empty slot):

\begin{equation}\label{eq:service_time}
D=(n_{\rm t}-1)\left(\frac{W}{2}\sigma + T_{\rm c}\right) + \frac{W}{2}\sigma+T_{\rm s},
\end{equation}

where $n_{\rm t}$ is the average number of attempts to successfully transmit a packet. The duration of an empty slot ($\sigma$), a successful transmission ($T_{\rm s}$) and a collision ($T_{\rm c}$) are computed as shown in Eq.~\ref{eq:Ts}. We have considered the same packet interframe spaces as in DCF for comparison purposes \cite{IEEE80211-IEEESTD1999}.

\begin{equation}\label{eq:Ts}
\sigma = T_{\rm s} = T_{\rm c} = \DIFS + T_{\rm fra} + \SIFS + T_{\rm ack},
\end{equation}


where $T_{\rm fra}$ and $T_{\rm ack}$ denote the times to transmit the frame and the acknowledgement, respectively. Considering also the DCF protocol, we compute $T_{\rm fra}$ as shown in Eq.~\ref{eq:Tfra} and $T_{\rm ack}$ as in Eq.~\ref{eq:Tack}.

\begin{equation}\label{eq:Tfra}
T_{\rm fra} = \frac{L_{\rm PLCPPre}+L_{\rm PLCPH}}{R_{\rm PHY}} + \frac{L_{\rm MACH}}{R_{\rm basic}} + \frac{L}{R_{\rm data}},
\end{equation}

\begin{equation}\label{eq:Tack}
T_{\rm ack} = \frac{L_{\rm PLCPPre}+L_{\rm PLCPH}}{R_{\rm PHY}} + \frac{L_{\rm ack}}{R_{\rm basic}},
\end{equation}

being $L_{\rm PLCPPre}$, $L_{\rm PLCPH}$, $L_{\rm MACH}$, $L_{\rm ack}$ and $L$ the length of the PLCP preamble, PLCP header, MAC header, acknowledgement and data payload, respectively, while $R_{\rm PHY}$, $R_{\rm basic}$ and $R_{\rm data}$ denote the physical, basic and data rates \cite{IEEE80211-IEEESTD1999}.

Under the decoupling assumption with an infinite number of retries, the average number of attempts to transmit a frame ($n_{\rm t}$) is computed as shown in Eq.~\ref{eq:nt}.

\begin{equation}\label{eq:nt}
 n_{\rm t}=\frac{1}{1-p},
\end{equation}

where the conditional collision probability ($p$) is obtained as the complementary of having at least one of the other $n-1$ nodes transmitting a frame in the same slot (Eq.~\ref{eq:p}), with $\tau$ denoting the attempt rate of a node. 

\begin{equation}\label{eq:p}
 p=1-(1-\tau)^{n-1}
\end{equation}

We view the attempt rate as a regenerative process, where the renewal events are when the MAC begins processing a new frame. Thus, we apply the renewal reward theorem (Eq.~\ref{eq:tau}). 

\begin{equation}\label{eq:tau}
 \tau=\frac{n_{\rm t}}{n_{\rm t}\left(\frac{W}{2}+1\right)+I}
\end{equation}

The term $I$ in Eq.~\ref{eq:tau} accounts for the number of slots in idle state (when there is no packet waiting in the queue for transmission) and is computed as the probability of having an empty queue over the probability of a packet arrival in a slot. Considering an M/M/1 queue, we then compute $I$ as in Eq.~\ref{eq:i}.

\begin{equation}\label{eq:i}
 I=\frac{1-\rho}{1-e^{-\lambda \sigma}}
\end{equation}

Finally, we obtain the throughput as:

\begin{equation}\label{eq:S}
S =  \rho \frac{L}{D}
\end{equation}

\section{Renewal Reward Analysis - DCF}\label{appendix:dcf}

The analysis used for DCF is similar to the one presented for Aloha in Appendix \ref{appendix:aloha}. The service time now takes into account the average number of slots waiting for the backoff to expire ($E[w]$). Moreover, we need to consider that the duration of a backoff slot is also no longer $\sigma$ but that it depends on the transmissions of the other nodes in the network. The new expression for the service time is:

\begin{equation}\label{eq:service_time_dcf}
D=(n_{\rm t}-1)(E[w]\alpha + T_{\rm c}) + E[w]\alpha+T_{\rm s},
\end{equation}

where $\alpha$ is the average slot duration while the node is in backoff and the transmission attempt probability changes as:

\begin{equation}\label{eq:tau_dcf}
 \tau=\frac{n_{\rm t}}{n_{\rm t}\left(E[w]+1\right)+I},
\end{equation}

where $I$ now considers the average slot duration ($\alpha$):

\begin{equation}\label{eq:i_dcf}
 I=\frac{1-\rho}{1-e^{-\lambda \alpha}}.
\end{equation}

The average slot duration while the node is in backoff is derived depending on the type of slot that is overheard (Eq.~\ref{eq:alpha_dcf}). A slot can be empty if no other node transmits (that occurs with $p_{\rm e}$ probability) and, in such a case, its duration is $\sigma$ (defined in \cite{IEEE80211-IEEESTD1999}). Otherwise, it can be occupied due to a successful transmission (that happens with probability $p_{\rm s}$) or a collision (that occurs with $p_{\rm c}$ probability), with durations $T_{\rm s}$ and $T_{\rm c}$, respectively. 

\begin{equation}\label{eq:alpha_dcf}
 \alpha=p_{\rm s}T_{\rm s} + p_{\rm c}T_{\rm c} + p_{\rm e}\sigma
\end{equation}

Probabilities $p_{\rm s}$, $p_{\rm e}$ and $p_{\rm c}$ are obtained as follows:

\begin{align}\label{eq:p_others_dcf}
  p_{\rm s}=(n-1)\tau (1-\tau)^{n-2},\nonumber\\
  p_{\rm e}=(1-\tau)^{n-1},\nonumber\\
  p_{\rm c}=1 - p_{\rm s} - p_{\rm e}.
\end{align}

Finally, the average number of backoff slots can be computed as shown in Eq.~\ref{eq:ew_dcf} derived in \cite{bianchi05}.

\begin{equation}\label{eq:ew_dcf}
 E[w] = \frac{1-p-p(2p)^m}{1-2p}\frac{W}{2}-\frac{1}{2}
\end{equation}

\section*{Acknowledgments}

This work has been partially supported by the Science Foundation Ireland grant 08/SRC/I1403 and 07/SK/I1216a.

\ifCLASSOPTIONcaptionsoff
  \newpage
\fi

\bibliographystyle{IEEEtran}
\bibliography{library}

\vspace{0.25cm}

\begin{wrapfigure}{l}{30mm}
    \includegraphics[width=30mm]{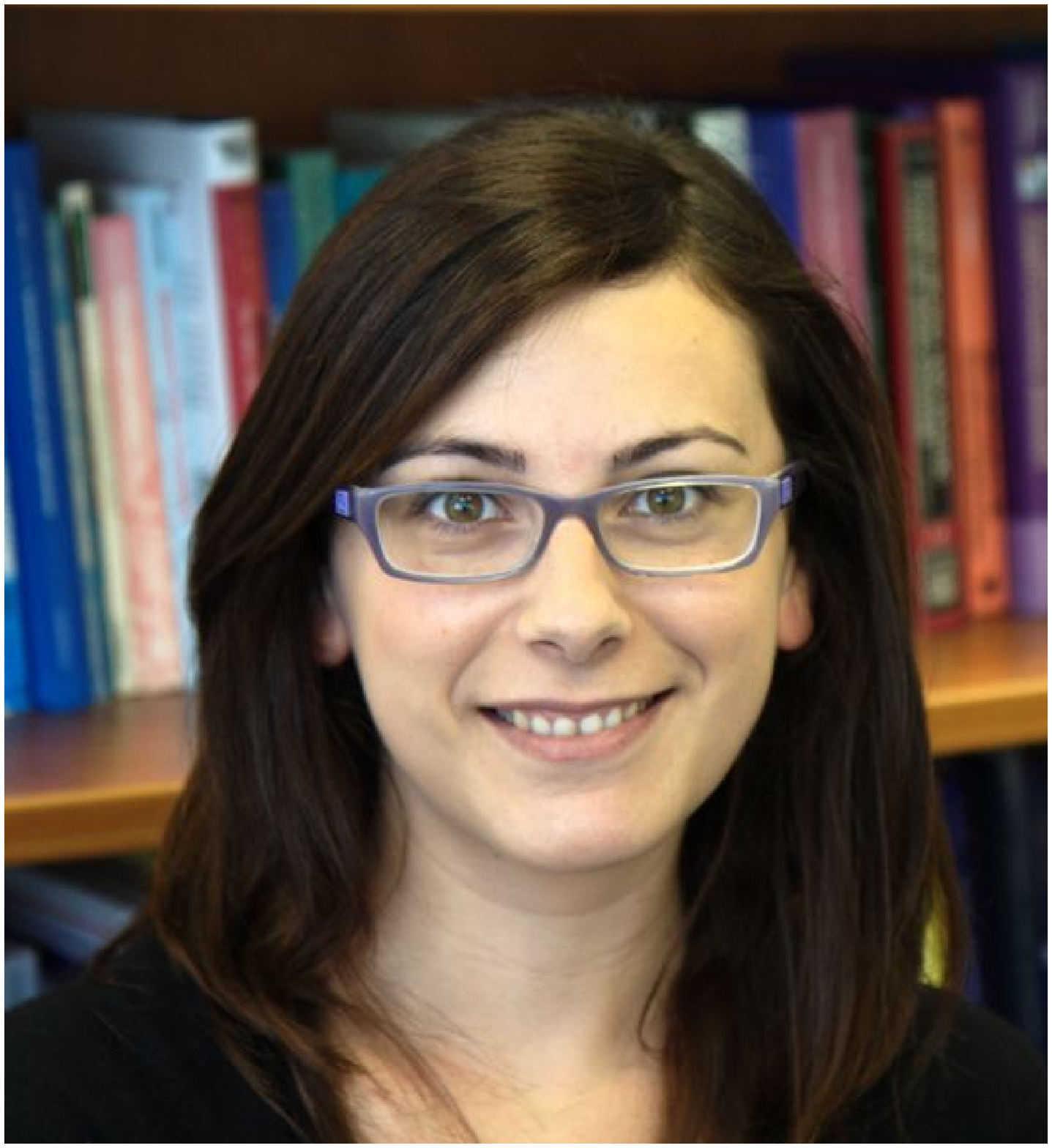}
\end{wrapfigure}
\textbf{Cristina Cano} obtained the Telecommunications Engineering Degree at the Universitat Politecnica de Catalunya (UPC) in February 2006. Then, she received her M.Sc. (2007) and Ph.D. (2011) on Information, Communication and Audiovisual Media Technologies from the Universitat Pompeu Fabra (UPF). Since July 2012, she has been working as a research fellow at the Hamilton Institute (NUIM) in wireless networks, sensor networks, power line communications and MAC layer design.\\

\begin{wrapfigure}{l}{30mm}
    \includegraphics[width=30mm]{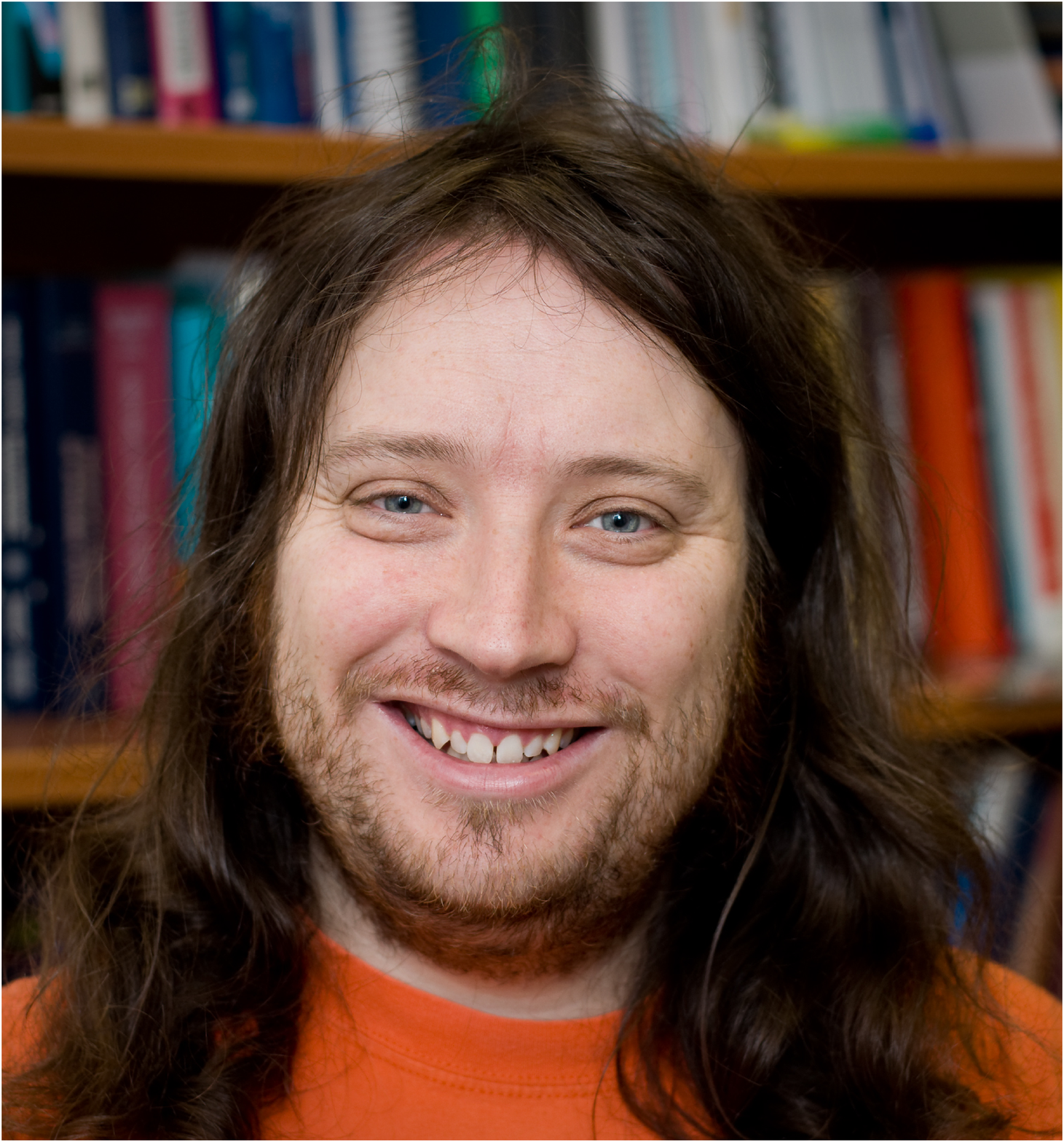}
\end{wrapfigure}
\textbf{David Malone} received B.A. (mod), M.Sc. and Ph.D. degrees in mathematics from Trinity College Dublin. During his time as a postgraduate, he became a member of the FreeBSD development team. He is currently a SFI Stokes Lecturer at the Hamilton Institute, NUI Maynooth. His interests include mathematics of networks, network measurement, IPv6 and systems administration. He is a coauthor of \emph{O'Reilly's IPv6 Network Administration}.

\end{document}